\documentclass[10pt]{amsart}
\usepackage[utf8]{inputenc}
\usepackage[top=1.5cm, bottom=1.5cm, left=2.0cm, right=2.0cm]{geometry}
\usepackage[utf8]{inputenc}
\usepackage{amsthm,amsfonts,amsmath,bookmark,appendix}
\usepackage{enumerate,graphicx,pdfpages,hyperref,dsfont,tabularx}
\usepackage{bbm}
\usepackage{amssymb}
\usepackage{bm}

\DeclareMathOperator{\supp}{supp}
\DeclareMathOperator{\Tr}{Tr}

\DeclareMathOperator{\eps}{\varepsilon}
\DeclareMathOperator{\0}{|0\rangle}
\def\ra{\rangle}
\def\la{\langle}
\DeclareMathOperator{\1}{\mathbbm{1}}

\def\R {\mathbb{R}}
\def\Z {\mathbb{Z}}

\def\C {\mathbb{C}}
\def\U {\mathbb{U}}
\def\F {\mathcal{F}}
\def\T {\mathbb{T}}
\def\N {\mathbb{N}}
\def\U {\mathbb{U}}

\def\H{\gH}

\def\R {\mathbb{R}}
\def\C {\mathbb{C}}
\def\U {\mathbb{U}}

\def\H{\mathcal{H}}
\def\NN{\mathcal{N}}

\def\dv{\text{div}}
\def\nn{\text{nn}}
\def\Neu{\text{Neu}}
\def\epsG{\varepsilon_{\text{gap}}^{\Neu}}
\def\epsN{\varepsilon^{\Neu}}
\def\a{\mathfrak{a}}
\def\Lh{\widehat \Lambda}
\def\Id{\text{Id}}
\def\GC{\text{GC}}
\def\sp{\text{special}}
\def\Per{\text{Per}}
\def\rank{\text{rank}}
\def\gap{\text{gap}}

\newtheorem{theorem}{Theorem}[section]
\newtheorem{prop}[theorem]{Proposition}
\newtheorem{cor}[theorem]{Corollary}
\newtheorem{lemma}[theorem]{Lemma}
\newtheorem{defi}[theorem]{Definition}

\theoremstyle{remark}
\newtheorem{remark}[theorem]{Remark}

\numberwithin{equation}{section}

\title{Ground state energy of the dilute Bose-Hubbard gas on Bravais lattices}

\author[N. Mokrza\'nski]{Norbert Mokrza\'nski}
\address{Department of Mathematical Methods in Physics, Faculty of Physics, University of Warsaw,  Pasteura 5, 02-093 Warszawa, Poland}
\email{norbert.mokrzanski@fuw.edu.pl} 

\author[M. Napi\'orkowski]{Marcin Napi\'orkowski}
\address{Department of Mathematical Methods in Physics, Faculty of Physics, University of Warsaw,  Pasteura 5, 02-093 Warszawa, Poland}
\email{marcin.napiorkowski@fuw.edu.pl} 

\author[J. Wojtkiewicz]{Jacek Wojtkiewicz}
\address{Department of Mathematical Methods in Physics, Faculty of Physics, University of Warsaw,  Pasteura 5, 02-093 Warszawa, Poland}
\email{jacek.wojtkiewicz@fuw.edu.pl} 

\begin{document}

\begin{abstract}
We study interacting bosons on a three–dimensional Bravais lattice with positive hopping amplitudes and on-site repulsive interactions. We prove that, in the dilute limit $\rho\to 0$, the ground state energy density satisfies 
$$e_0(\rho)=4\pi \a \rho^2 \big(1+O(\rho^{1/6})\big),$$
where $\a$ is the lattice scattering length defined through the corresponding two–body problem. This establishes the analogue of the Dyson and Lieb–Yngvason theorems for the Bose-Hubbard gas. Our result shows that the leading-order energy is universal: although the lattice geometry affects the microscopic dispersion relation, it enters the leading order asymptotics only through the scattering length. In particular, it is independent of other features of the underlying Bravais lattice.
\end{abstract}
\maketitle


\section{Introduction}

Understanding the ground state energy of interacting quantum many–body systems is a central problem in mathematical physics. Although a complete description is generally out of reach, rigorous results can be obtained in suitable asymptotic regimes. One particularly tractable regime is the dilute limit, where the particle density $\rho$ is sufficiently small compared to the interaction scale. In this setting, Bose gases exhibit a remarkable universality: to leading order, the ground state energy depends only on a single effective parameter, the two–body scattering length $\a$.

For three–dimensional continuum Bose gases with repulsive interactions, this is seen in the so-called Lee-Huang-Yang-Wu \cite{LHY,Wu} formula
$$
e_0(\rho)=
4\pi \a \rho^2
\Bigl(1+\frac{128}{15\sqrt{\pi}}(\rho \a^3)^{1/2}+8(\frac{4\pi}{3}-\sqrt{3})\rho \a^3 \ln(\rho \a^3)+\dots\Bigr),
$$
which captures in the dilute regime $\rho \a^3 \to 0$ the correct ground state energy per unit volume,
up to corrections of order $\a \rho^2 (\rho \a^3)$, which are expected to no longer be universal in $\a$. The leading-order term was rigorously established by Dyson \cite{Dyson} as an upper bound and, over 40 years later, by Lieb and Yngvason \cite{LiebYngvason} as  a lower bound (see also \cite{Yin-10}).  An upper bound matching the second order term was established by Yau and Yin \cite{YauYin} (see also \cite{ErSchYao,BasCenSch,Bastietal24}), while a lower bound finishing the proof of the Lee-Huang-Yang conjecture was established by Fournais and Solovej in \cite{FourSol1} (see also \cite{FourSol2}). In \cite{Haber1, Haber2}, a new, simpler proof that also establishes the free energy expansion in the positive temperature case was given (see also \cite{Sei-08,Yin-10a,BasBocCenDeu-25}). Finally, recently, Brooks, Oldenbrug, Saint Aubin and Schlein \cite{Brooksetal-25} established for the first time an upper bound that includes the third order term (so-called Wu term).

A natural question is whether this universality persists in discrete settings. Bose gases realized in optical lattices are described by lattice Hamiltonians, most prominently by the Bose–Hubbard model \cite{GerKnol,FWGF}, which has become a standard effective model for interacting bosons and has been extensively studied in both theory and experiment. Such systems arise on a variety of lattice geometries \cite{Petsas-94} beyond the simple cubic case, which motivates the consideration of general Bravais lattices. In this setting the single–particle dispersion and the associated low–energy kinematics depend strongly on the geometry and hopping amplitudes of the underlying lattice. In contrast to the continuum, both the interaction and the lattice structure influence the two–body problem, and it is not a priori clear whether the leading-order energy retains a universal form independent of the microscopic details.

In this work we show that such universality indeed survives on lattices as far as the leading order term of the energy is concerned. We consider interacting bosons on an arbitrary three–dimensional Bravais lattice with positive hopping amplitudes and on-site repulsive interactions and prove that, in the dilute limit 
$$e_0(\rho)=4\pi \a \rho^2 (1+O(\rho^{1/6}))$$
where $\a$ denotes the lattice scattering length (cf. Appendix \ref{sec:Scattering}). Thus all microscopic information — both the interaction strength and the lattice geometry — is absorbed into this effective parameter, and the leading-order energy is independent of other details of the underlying  lattice. This provides a discrete analogue of Dyson and Lieb–Yngvason theorems for the Bose–Hubbard gas.

It is worth emphasizing that this universality is specific to the leading-order term. While in the continuum the second and (expectedly) third order terms also exhibit a universal structure depending only on the scattering length, in the lattice setting one expects higher-order terms to depend explicitly on the single-particle dispersion and hence on the geometry of the underlying lattice. In particular, beyond order $\a\rho^2$
the energy is not determined solely by the scattering length. Our result therefore identifies the precise regime in which lattice effects are completely absorbed into this effective parameter. We expect that the same leading-order universality holds for more general short-range lattice potentials.

Related universality results have been obtained for fermionic systems, both in the continuum \cite{LieSeiSol-05,Falconieetal-21,Giacomelietalupper-24,Giacomelietallower-25,ChenWuZhang-25} and on the cubic lattice (with nearest neighbor hopping) \cite{Giu,SeirYin}, where the leading-order energy is again determined by an appropriate scattering parameter. The available lattice proofs for fermions employ techniques that do not readily transfer to bosons. In particular, arguments based on Dyson-type lemmas  have no direct counterparts. Therefore, in order to prove our main result, we rely on techniques that have been developed more recently in the context of continuous bosonic systems. For the lower bound we use a localization method coupled with Bogoliubov theory \cite{BoccatoSeiringer,OptimalRate,ChongLiangNam}. To make it work we need to develop estimates for the eigenvalues of Neumann Laplacians on general Bravais lattices. In fact, these bounds lead to the relative error of order $O(\rho^{1/6})$ in the lower bound. The upper bound is an adaptation of the argument in \cite{ErSchYao} that allows to include lattice dispersion relations which are not radial. 

The remainder of the paper is organized as follows. In Section~\ref{sec:Model} we introduce the lattice framework and state the main result. The upper bound is obtained via suitable grand-canonical trial states and the equivalence of ensembles. This is presented in Section~\ref{sec:TrialState}. The corresponding lower bound is proved in Section~\ref{sec:Lower_bound}. Several auxiliary results (including the discussion on the Bravais lattices and the lattice scattering length) are collected in the Appendices.


\section{The model and the main result}
\label{sec:Model}

We start by presenting basic definitions and objects of interest. We refer to Appendix \ref{sec:Fourier} for the discussion concerning those aspects.

Our main object of interest is a three dimensional (monoatomic) crystal lattice. To define it we need to specify the underlying Bravais lattice and the neighborhood relation between the points on the lattice.

To this end we first fix three linearly independent vectors $a_1$, $a_2$, $a_3 \in \R^3$ and a matrix $A$ composed of those vectors as columns. We consider a Bravais lattice $\Lambda$ defined as
\begin{equation}
\label{def:Bravais_lattice}
\Lambda = A\Z^3 = \left\{m_1 a_1 + m_2 a_2 + m_3 a_3 \colon m_1,m_2,m_3 \in \Z\right\}
\end{equation}
In this context vectors $a_i$ are called the primitive (translation) vectors of the lattice $\Lambda$.

For a given even number $L \in 2\N$ we consider a finite version of the Bravais lattice \eqref{def:Bravais_lattice} of size $L$, denoted $\Lambda_L$ and defined as
\begin{equation}
\label{def:finite_lattice}
\begin{split}
\Lambda_L &= A\left(\Z \cap [-L/2,L/2]\right)^3 
\\&= \left\{m_1 a_1 + m_2 a_2 + m_3 a_3 \colon m_1,m_2,m_3 = -\frac{L}{2}, -\frac{L}{2} + 1, \dots, \frac{L}{2} - 1, \frac{L}{2}\right\}.    
\end{split}
\end{equation}
We equip $\Lambda_L$ with periodic boundary condition (i.e $\Lambda_L \simeq A\Z^3/A((L+1)\Z)^3$ and with the standard counting measure, hence we can define the (one particle) Hilbert space $\H_L$ of a particle on the lattice $\Lambda_L$ as
$$\H_L = L^2\left(\Lambda_L\right)$$
with the inner product
$$\la \psi, \varphi\ra_{\H_L} = \sum_{x \in \Lambda_L} \overline{\psi(x)} \varphi(x).$$
For a given natural number $N$ we also define $N$-particle bosonic Hilbert space $\H_L^N$ as
$$\H_L^N = \bigotimes_{\text{sym}}^N \H_L,$$
i.e. the functions of $N$ variables $x_1$, $x_2$, \dots, $x_N \in \Lambda_L$ symmetric under the permutations of those variables. The inner product on this space is defined for simple tensors as
$$\left \la \bigotimes_{j=1}^N \psi_j, \bigotimes_{j=1}^N\varphi_j \right \ra_{\H_L^N} = \prod_{j=1}^N \la \psi_j, \varphi_j \ra_{\H_L},$$
which can be extended to the whole $\H_L^N$ by linearity. 

Now we will define the neighbor relation on $\Lambda$ and $\Lambda_L$. Let $D$ be a set of all "positive directions"
\begin{equation*}
D = \{m_1 a_1 + m_2a_2 + m_3a_3 \in \Lambda \setminus\{0\} \colon \text{ the first non-zero } m_j \text{ is positive}\}.
\end{equation*}
Note that $D \cup (-D) = \Lambda \setminus\{0\}$ and $D \cap (-D) = \emptyset$. To each direction $v \in D$ and its reverse direction $(-v)$ we will assign a weight $t(v) = t(-v) \ge 0$. The neighborhood relation on $\Lambda$ is defined as
\begin{equation}
\label{neighbor_relation_def}
x \sim y \Longleftrightarrow t(y - x) > 0.
\end{equation}
This relation is symmetric and equips both the infinite lattice $\Lambda$ and the finite lattice $\Lambda_L$ with the weighted graph structure, where in the latter $y-x$ is understood in the sense of periodic boundary condition, i.e. as the element of the $A\Z^3/A((L+1)\Z)^3$ group.

We will make two additional assumptions. The first one is that
\begin{equation}
\label{finite_range_assumption}
    \#\{v \in D \colon t(v) \ne 0\} < \infty,
\end{equation}
meaning that we only consider a finite distance hopping. This assumption is satisfied in most commonly encountered crystal systems in physics. For the future purposes we will also define a parameter $R_0(t)$ called the hopping length as
\begin{equation}
\label{def:hopping_length}
R_0(t) = \min \{L \in 2\N \colon \forall_{x \sim 0} \; x \in \Lambda_L\},
\end{equation}
that is the smallest $L$ such that all the neighbors of point $x=0$ in the sense of \eqref{neighbor_relation_def} belong to $\Lambda_L$. In the upcoming proofs we will consider only $L \ge R_0(t)$ as this condition will allow us to capture all the possible hoppings within the finite volume.

To state the second assumption we will first denote
\begin{equation}
\label{D_1_def}
D_1 := \{a_1, a_2, a_3\} \subset D.
\end{equation}
We will assume
\begin{equation}
\label{t_1_assumption}
t(v) \ne 0 \text{ for } v \in D_1
\end{equation}
that is hopping along the primitive vectors of the lattice $\Lambda$ is allowed.

The Bose-Hubbard Hamiltonian $H_{N,L}$ of the $N$ particle system, acting on $\H_L^N$, is given by
$$H_{N,L} = -\sum_{i = 1}^N\Delta_{i} + U\sum_{i<j}^N\delta_{x_i,x_j},$$
where  the first term is the kinetic energy. Here $\Delta_{x_i}$ denotes the lattice (weighted graph) Laplace operator acting on the $i$-th particle:
$$\Delta_i = \Id \otimes \dots \otimes \Delta \otimes \dots \otimes \Id,$$
where $\Id$ is the identity operator on $L^2(\Lambda_L)$ and $\Delta$ is a single particle Laplacian standing on the $i$-th position. We can write the action of $\Delta$ explicitly: for $u \in L^2(\Lambda_L)$
\begin{equation}
\label{def:discrete_laplacian}
-\Delta u (x) = \sum_{y \sim x} t(y-x)\big(u(x) -u(y)\big) = \sum_{v \in D} t(v) \big(2 u(x) - u(x+v) - u(x-v)\big), 
\end{equation}
where $x,y \in \Lambda_L$ and $x \sim y$ denotes the neighborhood relation. Here $y - x$ is once again understood in the sense of periodic boundary condition. The second term in the Hamiltonian is the interaction energy with $U > 0$ (i.e. the interaction is repulsive).

The ground state energy $E_0(N,L)$ of the system is defined by
\begin{equation}
\label{def:GSE}
E_0(N,L) =  \inf_{\substack{\psi \in \H_L^N \\ \|\psi\|=1}} \la \psi, H_{N,L} \psi \ra_{\H_L^N}.   
\end{equation}
The inner product above is called the expectation value of the operator $H_N$ in the state $\psi$. In general, the expectation value of some operator $T$ acting on the Hilbert space $\H$ in the state (i.e. normalized vector) $\psi \in \H$ is defined as
\begin{equation}
\label{def:expectation}
 \la T \ra_{\psi} = \la \psi, T \psi \ra_{\H}.
\end{equation}
We will use this notation when there will be no ambiguity on which Hilbert space this expectation is evaluated.

We are interested in the ground state energy per unit volume, i.e.
\begin{equation}
\label{def:ThermodynamicLimit}
e_0(\rho) = \lim_{\substack{N \to \infty \\ L \to \infty \\ N/|\Lambda_L| \to \rho}}\frac{E_0(N,L)}{|\Lambda_L|}.
\end{equation}
Existence of this limit (under some more general assumptions and even in some broader setting) is known, we refer e.g. to \cite{Ruelle} for the details. It is also known that $e(\rho)$ is a continuous (up to the boundary $\rho = 0$) and convex function of $\rho$. 

In order to state the main theorems we introduce \textit{the scattering length} of the potential which we will denote $\a$  (see Appendix~\ref{sec:Scattering} for more details). For the on-site interaction potential that we are dealing with it is defined as
\begin{equation}
\label{def:scattering_len}
8\pi \a = \frac{U}{U\gamma + 1},
\end{equation}
with
$$\gamma = \frac12|\Lh|^{-1}\int_{\Lh} \frac{dp}{\eps(p)}$$
and $\eps(p)$ being the lattice dispersion relation, given by
\begin{equation}
\label{def:eps}
\eps(p) = \sum_{v \in D} 2t(v)\big(1 - \cos (v \cdot p)\big) = 4\sum_{v \in D} t(v) \sin^2\left(\frac{v \cdot p}{2}\right), \quad p = (p_1,p_2,p_3) \in \Lh.   
\end{equation}
Here $\Lh$ is the Brillouin zone of the lattice $\Lambda$
$$\Lh = B\T^3 = \left\{\tau_1 b_1 + \tau_2 b_2 + \tau_3 b_3 \colon -\frac12 \le \tau_i < \frac12 \right\} \text{ with periodic boundary conditions},$$
where $\T^3 = [-\frac12,\frac12)^3$ is a three dimensional unit torus (this identification also allows to identify the Haar measure $dp$ in the integral as the Lebesgue measure), $|\Lh|$ denotes the measure of this set and $B$ is a matrix composed of columns $b_1$, $b_2$, $b_3$ satisfying
$$a_i \cdot b_j = 2\pi \delta_{i,j},$$
that is these are the primitive vectors of the reciprocal lattice $\Lambda^*$. We refer to Appendix~\ref{sec:Fourier} for a more detailed discussion. The expression $\eps(p)$ can be seen as the eigenvalue corresponding to the eigenfunction
$$\chi_p(x) = e^{ip\cdot x}, \quad x \in \Lambda$$
of the discrete Laplacian defined in \eqref{def:discrete_laplacian}. The sum in \eqref{def:eps} is finite due to the assumption \eqref{finite_range_assumption}. The formula \eqref{def:scattering_len} is derived explicitly in Appendix \ref{sec:Scattering}.

We will now check the assumption \eqref{t_1_assumption} implies that for $p \in \Lh$ we have the estimate
\begin{equation}
\label{eps_bound}
\eps(p) \ge c|p|^2
\end{equation}
with $|p|$ being the Euclidean norm of the vector $p \in \R^3$ and with the constant $c$ independent of $p$. Let us denote
$$t_{\min} := \min_{v \in D_1} t(v) > 0$$
and
$$b_{\max} := \max_{j=1,2,3} |b_j|.$$
Using inequality $\sin x \ge \frac{2}{\pi}x$ valid on the interval $[0,\pi/2]$ for $p = \tau_1 b_1 + \tau_2 b_2 + \tau_3 b_3 \in \Lh$ we have
$$\eps(p) \ge 4 \sum_{v \in D_1} t(v) \sin^2\left(\frac{v \cdot p}{2}\right) \ge 4 t_{\min}\sum_{j = 1}^3 \sin^2\left(\pi |\tau_j|\right) \ge 16 t_{\min} \sum_{j=1}^3 \tau_j^2.$$
On the other hand, by Cauchy-Schwarz inequality
\begin{align*}
|p|^2 = |\tau_1 b_1 + \tau_2 b_2 + \tau_3 b_3|^2  \le \sum_{j=1}^3 \tau_j^2 |b_j|^2 + 2\sum_{1 \le i<j \le 3} \tau_i \tau_j |b_i| |b_j| = \left(\sum_{j=1}^3 \tau_j |b_j|\right)^2 \le 3b_{\max}^2 \sum_{j=1}^3 \tau_j^2.
\end{align*}
We conclude that \eqref{eps_bound} holds with
$$c = \frac{16 t_{\min}}{3 b_{\max}^2}.$$

Inequality \eqref{eps_bound} in particular implies that $\gamma$ is well defined as the function $1/\eps(p)$ is integrable near zero. Moreover, as we will see in the proofs of the following propositions, this inequality will be crucial for obtaining the desired results. The assumption \eqref{t_1_assumption} itself can be changed in such a way that \eqref{eps_bound} still holds true, for example assuming that some certain other hopping constants are non-zero (this however would require some modifications in the proofs). For the purpose of this paper we will stick to \eqref{finite_range_assumption} as this is the simplest case when \eqref{eps_bound} holds.

The main theorem that we prove is the following
\begin{theorem}
\label{main_thm}
In the setting as above, in particular with assumptions \eqref{finite_range_assumption} and \eqref{t_1_assumption}, we have
$$e_0(\rho) = 4\pi \a \rho^2 \left(1 + O\left(\rho^{1/6}\right)_{\rho \to 0}\right),$$
\end{theorem}

We will prove Theorem \ref{main_thm} by proving the upper and the lower bound separately. In fact the upper bound (see below) provides a better error estimate that the one coming from the lower bound.

 As mentioned in the introduction,  the proof of the upper bound will be based on \cite{ErSchYao}, adapted to the lattice setting. The main point of this approach is that instead of constructing a sequence of states on $\H_L^N$, i.e. states with fixed number of particles, we will work in the grand canonical setting. More precisely, we will construct a sequence of states $\{\Psi_{L,N}\}_{L,N}$ on the Fock spaces
$$\F_L :=\F(\H_L) = \bigoplus_{n=0}^\infty \H_L^n, \quad (\H_L^0 = \C)$$
with each having fixed average number of particles $\la \NN\ra_{\Psi_{N,L}} = N$. We will denote the grand canonical Hamiltonian (i.e. the second quantization of $H_{N,L}$) as $H^\GC_L$. This operator acts on the Fock space $\F_L$, for $\Psi = (\Psi^{(n)})_{n \in \N}\in \F_L$ its action is given by
\begin{equation}
\label{def:gc_ham}
(H^\GC_L \Psi)^{(n)} = H_{n,L} \Psi^{(n)}. 
\end{equation}
We are going to prove the following result.
\begin{prop}
\label{main_thm_gc}
Let $\rho > 0$ be small enough. For any sequences $N \to \infty$, $L \to \infty$ with $\frac{N}{|\Lambda_L|} \to \rho$ there exists a sequence of trial states $\Psi_L \in \F_L$ with $\la \NN \ra_{\Psi_L} = N$ such that
$$\lim_{\substack{L \to \infty \\ N \to \infty \\ N/|\Lambda_L| \to \rho}} \frac{\la H^\GC_L \ra_{\Psi_L}}{|\Lambda_L|} = 4 \pi \a \rho^2\left(1 + O(\rho^{1/2})\right).$$
\end{prop}
The upper bound (with the same error term as above) follows then from the variational principle and the equivalence of ensembles. The adaptation of this well known argument for the discrete setting will be presented at the end of Section~\ref{sec:TrialState}, after the proof of Proposition \ref{main_thm_gc}. In particular, when considering the cubic lattice with nearest neighbor hopping one can reproduce the weak-coupling Lee-Huang-Yang upper bound analogous as in \cite{ErSchYao} (see also \cite{NapReuSol-18}).

As for the lower bound we will use the method of dividing the large (thermodynamic) lattice $\Lambda_L$ into smaller sub-lattices $\Lambda_\ell$ with a properly chosen length scale $\ell$. In the standard proof of the corresponding bound for the continuous (i.e. non-discrete) setting presented e.g. in \cite[Chapter 2.2]{MathBEC} the length scale $\ell$ is chosen is such a way that $\rho^{1/3} \ll \ell \ll \rho^{-1/2}$, which allows to effectively use the Dyson lemma and obtain the desired result. Since we were unable to prove the discrete analogue of this lemma that would be of use to us, we propose a different approach and choose $\ell \sim \rho^{-1/2}$, which is commonly known as the Gross-Pitaevskii length scale. Then we will use the method from \cite{OptimalRate} (also recently used in \cite{ChongLiangNam}) to get the operator inequality bounding $H_{n,\ell}$ from below for certain values of $n$ and $\ell$. In the end we will use the obtained bound and the method from \cite{BoccatoSeiringer} to get the following result.
\begin{prop}
\label{prop:lower_bound}
For small enough $\rho > 0$ the ground state energy density in the thermodynamic limit satisfies
$$e_0(\rho) \ge 4\pi \a \rho^2\left(1 - O(\rho^{1/6}) \right).$$
\end{prop}
Let us stress that the worse error term than in the upper bound is a consequence of spectral estimates that we derive for general Bravais lattices with general neighbor relations. For example, this error can be improved to be of order $O(\rho^{1/2} \ln \rho)$ if one considers cubic lattices with nearest neighbor hopping. The proof of Proposition \ref{prop:lower_bound} is given Section~\ref{sec:Lower_bound}. This will finish the proof of Theorem \ref{main_thm}.

In the rest of the paper we use the convention that $C$ denotes a generic constant (independent of relevant parameters) which may change from line to line. 


\section{The upper bound}
\label{sec:TrialState}
This section is devoted the proof of Proposition \ref{main_thm_gc} and the corresponding upper bound. As mentioned before, the idea of the proof will follow the one in \cite{ErSchYao}, but with some adaptation to the discrete setting. In particular, we will use methods that do not rely on the spherical symmetry of the system. The proof will be done in a few steps, many of them being by now standard in the field.

\subsection{Momentum representation of the Hamiltonian}

It will be convenient to rewrite the grand canonical Hamiltonian \eqref{def:gc_ham} in the formalism of creation and annihilation operators -  we refer to \cite{SolovejNotes} and \cite{Diag} for more details concerning second quantization and Bogoliubov transformations. We will also use notation from Appendix \ref{sec:Fourier}. 

We will fix $L \in 2\N$ satisfying $L \ge R_0(t)$ (recall definition \eqref{def:hopping_length}) and work within the momentum representation. Following Appendix \ref{sec:Fourier} (in particular section \ref{subsec:finite_lattice}) we denote
$$\Lh_L = \left\{\sum_{j=1}^d m_j \frac{b_j}{L+1} \colon m_j=-\frac{L}{2},-\frac{L}{2} + 1,\dots, \frac{L}{2} - 1, \frac{L}{2}\right\},$$
where $b_j$ are the primitive vectors of the reciprocal lattice $\Lambda^*$. For $p \in \Lh_L$ we denote by $a_p$ and $a_p^*$ annihilation and creation operators of a particle with a momentum $p \in \Lh_L$, that is
$$a_p = a(\chi_p), \quad a^*_p = a^*(\chi_p),$$
where
$$\chi_p(x) = \frac{1}{|\Lambda_L|^{1/2}}e^{ip \cdot x}, \quad x \in \Lambda_L.$$
Direct computation of the matrix elements of the one- and two-body operators in $H_{N,L}$ in the above basis yields
\begin{equation}
\label{Ham_mom}
H_L^\GC = \sum_{p} \eps(p) a^*_p a_p +\frac{U}{2 |\Lambda_L|}\sum_{p,q,k}  a_{p+k}^* a_{q-k}^* a_q a_p,  
\end{equation}
where indices $p,q,k$ run over $\Lh_L$ and $\eps(p)$ is defined analogously as in \eqref{def:eps} but only for discrete values of $p$:
\begin{equation}
\label{def:eps_finite}
\eps(p) = \sum_{v \in D} 2t(v)\big(1 - \cos (v \cdot p)\big), \quad p = (p_1,p_2,p_3) \in \Lh_L.  
\end{equation}
We also note that creation and annihilation operators satisfy standard canonical commutation relations
$$[a_p,a_k] = [a_p^*, a_k^*] = 0, \quad [a_p, a_k^*] = \delta_{p,k}.$$

\subsection{Construction of the states}

We proceed to the construction of the trial state $\Psi_L \in \F_L$ with $\la \NN \ra_{\Psi_L} = N$. As for this moment values of $N$ and $L$ are fixed, we will simplify notation and omit index $L$ in some of the objects (i.e. $\Psi_L = \Psi$ etc.).

Consider $N_0$ satisfying $0 \le N_0 \le N$ and for $p \in \Lh_L\setminus \{0\}$ consider a (finite) sequence of real numbers $(c_p) \subset \R$, such that $|c_p| < 1$ and $c_p = c_{-p}$. Define the state $\Psi$ as
\begin{equation}
\label{state_def}
\Psi =\left(e^{-N_0/2}\prod_{p \ne 0} (1-c_p^2)^{1/4} \right) e^{\frac12\sum_{p \ne 0} c_p a_p^* a_{-p}^* + \sqrt{N_0}a_0^*}\0,  
\end{equation}
where $p$ belongs to $\Lambda^*$. Here the exponent should be understood as a notation for the proper series expansion. One can recognize that $\Psi$ is a so-called Bogoliubov trial state, that is it is of the form
$$\Psi = W\U^*\0,$$
where
$$W = W\left(\sqrt{N_0/|\Lambda|}\right) = e^{\sqrt{N_0}(a_0^* - a_0)}$$
is the Weyl operator build upon constant function $\sqrt{N_0/|\Lambda|}$ and $\U$ is the Bogoliubov transformation given by
\begin{equation}
\label{Bog_transform}
\U = \exp \left(\sum_{p\in \Lh \setminus \{0\}} -\frac{\text{artanh } c_p}{2} \left(a^*_p a^*_{-p} - a_p a_{-p}\right)\right).   
\end{equation}
The state $\Psi$ is normalized and conserves momentum, meaning that 
\begin{equation}
\label{mom_cons1}
p \ne q \Rightarrow \la a_p^*a_q\ra_\Psi = 0 \qquad \text{and} \qquad 
p \ne -q \Rightarrow \la a_pa_q\ra_\Psi = \la a_p^*a_q^*\ra_\Psi= 0.  
\end{equation}

\subsection{Computation of the energy}
We will now find the energy $\la H_L^\GC \ra_\Psi$ of the system in the \mbox{state $\Psi$}. This is a well-known computation, so will only state the main steps. 

It follows from the properties of the Weyl and Bogoliubov transformations that
\begin{equation} \label{computation1}
\la a_0^*a_0\ra_\Psi = N_0, \qquad \la a^*_0a^*_0a_0a_0\ra_\Psi= \la a_0a_0\ra_\Psi = N_0^2
\end{equation}
and for $p \ne 0$
\begin{equation} \label{computation2}
\la a_p^*a_p\ra_\Psi = \frac{c_p^2}{1 - c_p^2}, \qquad \la a_p^*a_{-p}^*\ra_\Psi= \la a_pa_{-p}\ra_\Psi = \frac{c_p}{1 - c_p^2}.
\end{equation}
The same computation as in \cite[Appendix A.]{MokNap-24} leads, using \eqref{computation1} and \eqref{computation2}, to the following expression 
\begin{equation}
\label{energy}
\begin{split}
\la H_L^\GC \ra_\Psi &= \sum_{p \ne 0} \eps(p) \langle a^*_p a_p\rangle +\frac{U}{2 |\Lambda_L|}\sum_{p,q \ne 0}\left[ \langle a_{p}^* a_{-p}^*\rangle \langle a_q a_{-q} \rangle  + 2\langle a_{p}^* a_{p}\rangle \langle a^*_q a_{q} \rangle \right]
\\&+ \frac{UN_0}{2|\Lambda_L|} \sum_{p \ne 0}\left[ 2\langle a_p a_{-p} \rangle + 4\langle a_p^* a_{p} \rangle \right] + \frac{UN_0^2}{2|\Lambda_L|}
\\&= \sum_{p \ne 0} \frac{\eps(p) c_p^2}{1 - c_p^2} + \frac{U}{2|\Lambda_L|} \sum_{p,q \ne 0} \left[\frac{c_pc_q}{(1 - c_p^2)(1-c_q^2)} + \frac{2c_p^2c_q^2}{(1 - c_p^2)(1-c_q^2)}\right]
\\&+ \frac{UN_0}{|\Lambda_L|} \sum_{p \ne 0} \left[\frac{c_p}{1 - c_p^2} + \frac{2c_p^2}{1 - c_p^2}\right] + \frac{UN_0^2}{2|\Lambda_L|}.
\end{split}
\end{equation}

As mentioned in the statement of Proposition \ref{main_thm_gc} we will only consider states $\Psi \in \F$ with fixed  expectation value of particle numbers $N := \la \NN \ra_\Psi$ and later consider values of $N$ and $L$ such that
$$\frac{\la \NN \ra_\Psi}{|\Lambda_L|} = \frac{N}{L^3} \to \rho$$
when $N \to \infty$ and $L \to \infty$. By \eqref{computation1} and \eqref{computation2}  we have
$$\la \NN \ra_\Psi = N_0 + \sum_{p \ne 0}\frac{c_p^2}{1 - c_p^2}$$
hence for considered values of $N$ and $L$ we get
\begin{equation}
\label{rho_constraint}
\rho =  \frac{N_0}{|\Lambda_L|} + \frac1{|\Lambda_L|}\sum_{p \ne 0}\frac{c_p^2}{1 - c_p^2} + o(1)_{L \to \infty},
\end{equation}
where $o(1)_{L \to \infty}$ denotes the expression converging to zero as $L \to \infty$. We also observe that (as $\rho$ is fixed)
\begin{equation}
\label{rho^2_constraint}
\frac{1}{|\Lambda_L|^2}\left(N_0^2 + 2N_0\sum_{p \ne 0}\frac{c_p^2}{1 - c_p^2} + \left(\sum_{p \ne 0}\frac{c_p^2}{1 - c_p^2}\right)^2\right) = \rho^2 + o(1)_{L \to \infty}.  
\end{equation}
From now on we will assume that parameters $N_0$ and $c_p$ are chosen in such a way that \eqref{rho_constraint} is satisfied.

In the following part it will be convenient to rewrite the expectation value $\la H_L^{\GC} \ra_\Psi$ in \eqref{energy} to express it in terms of the total density $\rho$. The observations \eqref{rho_constraint}, \eqref{rho^2_constraint} and
$$\sum_{p,q \ne 0} \frac{c_p^2c_q^2}{(1 - c_p^2)(1- c_q^2)} = \left(\sum_{p\ne 0} \frac{c_p^2}{1 - c_p^2}\right)^2$$
imply that
\begin{equation*}
\begin{split}
\la H_L^\GC \ra_\Psi &= \sum_{p \ne 0} \frac{\eps(p) c_p^2}{1 - c_p^2} + \frac{U}{2 |\Lambda_L|} \sum_{p,q \ne 0} \frac{c_pc_q}{(1 - c_p^2)(1-c_q^2)} 
\\&+ U\left(\rho - \frac1{|\Lambda_L|}\sum_{q \ne 0}\frac{c_q^2}{1 - c_q^2} + o(1)_{L \to \infty}\right)\sum_{p \ne 0} \frac{c_p + c_p^2}{1 - c_p^2} + \frac{U}{2}|\Lambda_L|\left(\rho^2 +o(1)_{L \to \infty}\right),
\end{split}
\end{equation*}
which after rewriting gives
\begin{equation}
\label{energy_a}
\begin{split}
\la H_L^\GC \ra_\Psi &= \sum_{p \ne 0} \frac{\eps(p) c_p^2}{1 - c_p^2} + \frac{U\rho(c_p + c_p^2)}{1 - c_p^2}
\\&+ \frac{U}{2|\Lambda|}\sum_{p,q \ne 0} \frac{c_pc_q - 2c_q^2(c_p + c_p^2)}{(1 - c_p^2)(1-c_q^2)} + o(1)_{L \to \infty} \cdot \sum_{p \ne 0} \frac{c_p + c_p^2}{1 - c_p^2}
\\&+ \frac{U}{2}|\Lambda|\left(\rho^2 +o(1)_{L \to \infty}\right)
\end{split}
\end{equation}
We expect (and prove it in further steps) that with proper selection of coefficients $c_p$ the value of
$$\sum_{p,q \ne 0} \frac{c_p^2c_q^2}{(1 - c_p^2)(1-c_q^2)} = \left(\sum_{p \ne 0} \frac{c_p^2}{1 - c_p^2}\right)^2$$
will be negligible in the thermodynamic limit in the dilute regime. To this end we rewrite (using symmetry of summation with respect to indices $p$ and $q$)
\begin{align*}
\sum_{p,q \ne 0} \frac{c_pc_q - 2c_q^2(c_p + c_p^2)}{(1 - c_p^2)(1-c_q^2)} &= \sum_{p,q \ne 0}\frac{(c_p - c_p^2)(c_q -c_q^2)}{(1 - c_p^2)(1-c_q^2)} - 3\sum_{p,q \ne 0}\frac{c_p^2c_q^2}{(1 - c_p^2)(1-c_q^2)}
\\&= \left(\sum_{p \ne 0} \frac{c_p - c_p^2}{1 - c_p^2}\right)^2 - 3\left(\sum_{p \ne 0} \frac{c_p^2}{1 - c_p^2}\right)^2.
\end{align*}
With this result expression \eqref{energy_a}, after further simplifications $\frac{c_p - c_p^2}{1 - c_p^2} = \frac{c_p}{1 + c_p}$ and $\frac{c_p + c_p^2}{1 - c_p^2} = \frac{c_p}{1 - c_p}$ becomes
\begin{equation}
\begin{split}
\la H_L^\GC \ra_\Psi &= \sum_{p \ne 0} \frac{\eps(p) c_p^2}{1 - c_p^2} + \frac{U\rho c_p }{1 - c_p}
+ \frac{U}{2|\Lambda_L|}\left(\sum_{p \ne 0} \frac{c_p}{1 + c_p}\right)^2
\\&- \frac{3U}{2|\Lambda_L|}\left(\sum_{p \ne 0} \frac{c_p^2}{1 - c_p^2}\right)^2 + o(1)_{L \to \infty} \cdot \sum_{p \ne 0} \frac{c_p}{1 - c_p}
\\&+ \frac{U}{2}|\Lambda_L|\left(\rho^2 +o(1)_{L \to \infty}\right)
\end{split}
\end{equation}
As a final modification of this expression we will replace the squared term in the first line with a linear one at expense of some another negligible term in the low density limit. The main idea is to add and subtract a $\rho  w(0)$ term to every element of the sum (also recall that $w(0)$ is given in \eqref{useful}). We will denote
\begin{equation}
\label{s_p}
s_p = \frac{c_p}{1+c_p}   
\end{equation}
and for convenience we will additionally define $s_0 = 0$. Next we write
\begin{align*}
\left(\sum_{p \ne 0} s_p \right)^2 = \left(\sum_{p \in \Lh_L} s_p \right)^2 &= \left(\sum_{p \in \Lh_L} \big(s_p + \rho  w(0)\big)\right)^2 - 2|\Lambda_L|\rho w(0)\sum_{p \in \Lh_L} s_p  - |\Lambda_L|^2\rho^2 w(0)^2
\\&= \left(\sum_{p \in \Lh_L} \big(s_p + \rho  w(0)\big)\right)^2 - 2|\Lambda_L|\rho w(0)\sum_{p \ne 0} s_p  - |\Lambda_L|^2\rho^2 w(0)^2.
\end{align*}
Eventually we obtain the following expression for the energy in the state $\Psi$
\begin{equation}
\label{energy_final}
\begin{split}
\la H_L^\GC \ra_\Psi &= \sum_{p \ne 0}\left( \eps(p)\frac{c_p^2}{1 - c_p^2} + U\rho\frac{c_p}{1 - c_p} - U\rho w(0)\frac{c_p}{1 + c_p}\right) - \frac{U}{2}|\Lambda_L|\rho^2 w(0)^2
\\& + \frac{U}{2|\Lambda_L|}\left(\sum_{p \in \Lh_L} \big(s_p + \rho  w(0)\big)\right)^2- \frac{3U}{2|\Lambda_L|}\left(\sum_{p \ne 0}\frac{c_p^2}{1 - c_p^2}\right)^2 
\\&+ o(1)_{L \to \infty} \cdot \sum_{p \ne 0} \frac{c_p}{1 - c_p} + \frac{U}{2}|\Lambda_L|\left(\rho^2 + o(1)_{L \to \infty}\right).
\end{split}
\end{equation}

\subsection{Minimalization}

We proceed to the minimalization procedure. For a start, we are interested in term-by-term minimization of the sum in the first line of \eqref{energy_final}, that is the want to minimize:
\begin{equation}
\label{min_target}
\eps(p) \frac{c_p^2}{1 - c_p^2} + U\rho \frac{c_p + c_p^2}{1 - c_p^2} - U\rho w(0)\frac{c_p}{1 + c_p}.   
\end{equation}
As mentioned before, it will turn out that the quadratic terms (second line in \eqref{energy_final}) will be negligible with the selection of $c_p$ minimizing this expression. In a more concrete manner, we will start with proving the following Lemma.

\begin{lemma}
\label{min_prop}
The minimal value of \eqref{min_target} is
\begin{equation}
\label{energy_expression}
\frac12 \left(\sqrt{\left(\eps(p) + 2U\rho\right)\left(\eps(p) + 2U\rho w(0)\right)} - \eps(p) - U\rho\left(1 + w(0)\right)\right).
\end{equation}
The explicit values of $c_p$ can be recovered from relation \eqref{def:s_p}.
\end{lemma}

\begin{proof}
In order to minimize \eqref{min_target} it will be convenient to rewrite it in terms of variable $s_p$ introduced in \eqref{s_p}. The inverse relation is given by
\begin{equation}
\label{c_p_inverse}
c_p = \frac{s_p}{1 - s_p}   
\end{equation}
and since $c_p \in (-1,1)$ we have $s_p \in (-\infty, \frac12)$. Direct computation yields
$$\frac{c_p^2}{1 - c_p^2} = \frac{s_p^2}{1 - 2s_p}, \quad \quad \frac{c_p}{1 - c_p} = \frac{s_p}{1 - 2s_p},$$
so \eqref{min_target} expressed in terms of $s_p$ becomes
\begin{equation}
\eps(p) \frac{s_p^2}{1 - 2s_p} + U\rho \frac{s_p}{1 - 2s_p} - U\rho w(0)s_p.
\end{equation}
Now the minimalization problem reduces to finding minimum of the function
$$F(x) := A\frac{x^2}{1 - 2x} + B\frac{x}{1 - 2x} - Cx.$$
for $A,B,C>0$ on the domain $x < \frac12$. 
Computing the derivative
$$F'(x) = A\frac{2x(1 - 2x) + 2x^2}{(1 - 2x)^2} + B\frac{1 - 2x + 2x}{(1 - 2x)^2} - C = \frac{2Ax(1 - x) + B}{(1 - 2x)^2} - C$$
and equating it to zero shows, after some straightforward computation, that the minimal value of the function $F$ is attained at
$$x_0 = \frac12 - \frac12\left(1 + 2\frac{B - C}{A + 2C}\right)^{1/2}$$
and is equal to
$$F(x_0) = \frac12\left(\sqrt{(A+2B)(A+2C)} - (A + B + C)\right).$$
Going back to the original minimalization problem, we have
$$A = \eps(p), \quad B = U\rho, \quad C = U\rho w(0),$$
so the minimal value of the expression \eqref{min_target} is exactly as in \eqref{energy_expression} and is attained at
\begin{equation}
\label{def:s_p}
s_p = \frac12 - \frac12\left(\frac{\eps(p) + 2U\rho}{\eps(p) + 2U\rho w(0)}\right)^{1/2} = \frac12 - \frac12\left(1 + \frac{2U\rho(1 - w(0))}{\eps(p) + 2U\rho w(0)}\right)^{1/2}.  
\end{equation}
\end{proof}

Having minimized the local part of the energy \eqref{energy_final} now we will show that the remaining parts are negligible in the dilute limit $\rho \to 0$. 
\begin{lemma}
\label{lemma:remainder}
For $s_p$ chosen as in \eqref{def:s_p} (and respectively chosen $c_p$ as in \eqref{c_p_inverse}) we have asymptotic bounds
\begin{equation}
\label{negligible}
\frac{1}{|\Lambda_L|}\left[\frac{U}{2|\Lambda_L|}\left(\sum_{p\ne 0} \big(s_p + \rho  w(0)\big)\right)^2- \frac{3U}{2|\Lambda_L|}\left(\sum_{p\ne 0}\frac{c_p^2}{1 - c_p^2}\right)^2\right] \lesssim \rho^3.   
\end{equation}
Moreover for the $\left|\sum_{p \ne 0} \frac{c_p}{1 - c_p}\right|$ term we have
\begin{equation}
\label{negligible_N}
\sum_{p \ne 0} \frac{c_p}{1 - c_p} \lesssim |\Lambda_L|.
\end{equation}
The notation $x \lesssim y$ means $x \le c y$ for some constant $c > 0$ independent of $L$ and $\rho$.
\end{lemma}

\begin{proof}
We will start with analyzing the second term in \eqref{negligible}. First we check that
\begin{align*}
\frac{c_p^2}{1-c_p^2} = \frac{s_p^2}{1 - 2s_p} = \frac14 \left(1 + \frac{2U\rho(1 - w(0))}{\eps(p) + 2U\rho w(0)}\right)^{1/2} + \frac14 \left(1 + \frac{2U\rho(1 - w(0))}{\eps(p) + 2U\rho w(0)}\right)^{-1/2} - \frac12.
\end{align*}
Using the inequality (coming from the Taylor expansion)
$$(1 + x)^{1/2} + (1+x)^{-1/2} \le 2 + \frac{x^2}{4}$$
we can estimate
$$\frac{c_p^2}{1-c_p^2} \le \frac{U^2\rho^2(1 - w(0))^2}{4(\eps(p) + 2U\rho w(0))^2}.$$
Now we will deduce that
\begin{equation}
\label{reminder1}
\frac{1}{|\Lambda|}\sum_{p \ne 0}\frac{c_p^2}{1-c_p^2} \lesssim \rho^{3/2}  
\end{equation}
This result will follow from approximating the sum by the integral (as this is a Riemann sum of a continuous function on $\Lh$) and dividing the integration into regions with small momenta and momenta separated form zero (we use the fact that locally near $p=0$ the manifold $\Lh$ looks like a subset of the Euclidean space $\R^3$). For sufficiently large $L$ we have
\begin{equation}
\label{int_delta_split}
\begin{split}
\frac1{|\Lambda_L|}\sum_{p \ne 0} \frac{c_p^2}{1-c_p^2} &\lesssim |\Lh|^{-1}\int_{\Lh}\frac{U^2\rho^2(1 - w(0))^2}{(\eps(p) + 2U\rho w(0))^2}dp
\\&= |\Lh|^{-1}\int_{|p|\le p_0}\frac{U^2\rho^2(1 - w(0))^2}{(\eps(p) + 2U\rho w(0))^2}dp + |\Lh|^{-1} \int_{|p|>p_0}\frac{U^2\rho^2(1 - w(0))^2}{(\eps(p) + 2U\rho w(0))^2}dp.   
\end{split}
\end{equation}
Here we have chosen the same $p_0$ as used in \eqref{eps_bound}, in particular we have the bound $\eps(p) \ge c |p|^2$. We can simplify the upcoming bounds even more by noting
\begin{equation}
\label{gamma_ineq}
U\big(1-w(0)\big) \le \frac{1}{\gamma}. 
\end{equation}
Then, by using spherical coordinates, we get
\begin{equation}
\label{small_p}
\begin{split}
|\Lh|^{-1} \int_{|p|\le p_0}\frac{U^2\rho^2(1 - w(0))^2}{(\eps(p) + 2U\rho w(0))^2} dp &\le C \rho^2 \int_{|p|\le p_0}\frac{1}{(p^2 + 2U\rho w(0))^2} dp
\\&\le C \rho^2\int_{|p|\le p_0}\frac{1}{p^2(p^2 + 2U\rho w(0))} dp
\\&= C \rho^2 \int_0^{p_0} \frac{1}{(r^2 + 2U\rho w(0))}dr
\\&= C \rho^2 \frac{1}{\sqrt{2U\rho w(0)}} \arctan \left(\frac{p_0}{\sqrt{2U\rho w(0)}}\right)
\\&\le \frac{C}{\sqrt{Uw(0)}}\rho^{3/2}.
\end{split}
\end{equation}
The constant $C$ is dependent only on $c$ from \eqref{eps_bound}, $U$ and $|\Lh|$.

The second integral in \eqref{int_delta_split} can be estimated trivially as for $|p| > p_0$ we have $\eps(p) > c$ for some constant $c > 0$ (dependent only on the fixed $p_0$), so
$$|\Lh_L|^{-1}\int_{|p|>p_0}\frac{U^2\rho^2(1 - w(0))^2}{2(\eps(p) + 2U\rho w(0))^2} \le \frac{\rho^2}{2\gamma(c + 2U\rho w(0))^2} \le C\rho^2$$
Combining the above results inequality \eqref{reminder1} follows. We conclude that 
$$\frac{1}{|\Lambda_L|}\left(\sum_{p \ne 0}\frac{c_p^2}{1 - c_p^2}\right)^2 = |\Lambda_L| \left(\frac1{|\Lambda_L|}\sum_{p \ne 0}\frac{c_p^2}{1 - c_p^2}\right)^2 \lesssim  |\Lambda_L| \rho^3$$
hence in the thermodynamic limit $(L \to \infty, N \to \infty, N/|\Lambda_L| \to \rho)$ we have the asymptotics
$$\frac{1}{|\Lambda_L|^2}\left(\sum_{p \ne 0}\frac{c_p^2}{1 - c_p^2}\right)^2 \lesssim \rho^3,$$
which proves this term is indeed negligible.

Now we proceed to estimate the other term in \eqref{negligible}. Once again we are dealing with the continuous function on $\Lh$, hence for sufficiently large $L$ we can approximate the sum by the integral:
\begin{equation}
\label{sp_series}
\left|\frac{1}{|\Lambda_L|}\sum_{p \ne 0}\left(s_p + \rho  w(0)\right)\right|
\lesssim |\Lh|^{-1}\left|\int_{\Lh}\left(\frac12 - \frac12\left(1 + \frac{2U\rho(1 - w(0))}{\eps(p) + 2U\rho w(0)}\right)^{1/2} + \rho  w(0)\right)dp\right|
\end{equation}
Using equation \eqref{w_transform} we have
$$w(0) = \int_{\T^3} \widehat w(p) dp = \int_{\T^3}\frac{U(1 - w(0))}{2\eps(p)}dp$$
we can write
\begin{equation}
\label{w(0)_int}
\begin{split}
&\;|\Lh|^{-1}\left|\int_{\Lh}\left(\frac12 - \frac12\left(1 + \frac{2U\rho(1 - w(0))}{\eps(p) + 2U\rho w(0)}\right)^{1/2} + \rho  w(0)\right)dp\right|
\\&= |\Lh|^{-1}\left|\int_{\Lh}\left(\frac12 - \frac12\left(1 + \frac{2U\rho(1 - w(0))}{\eps(p) + 2U\rho w(0)}\right)^{1/2} +  \frac{U\rho(1-w(0))}{2\eps(p)}\right)dp\right|
\end{split}   
\end{equation}
Using the inequality
$$\sqrt{1+x} \ge 1 + \frac12 x - \frac14 x^2$$
we also have
$$\frac12 - \frac12\left(1 + \frac{2U\rho(1 - w(0))}{\eps(p) + 2U\rho w(0)}\right)^{1/2} \le -\frac{U\rho(1 - w(0))}{2(\eps(p) + 2U\rho w(0))} + \frac12\left(\frac{U\rho(1 - w(0))}{\eps(p) + 2U\rho w(0)}\right)^2.$$
Moreover
\begin{equation*}
-\frac{U\rho(1 - w(0))}{2(\eps(p) + 2U\rho w(0))} + \frac{U\rho(1-w(0))}{2\eps(p)} = \frac{U^2\rho^2 w(0)(1-w(0))}{\eps(p)(\eps(p) + 2U\rho w(0))},
\end{equation*}
hence, after some more straightforward estimates
$$\eqref{w(0)_int} \le |\Lh|^{-1}\int_{\Lh}\frac{U^2\rho^2 (1-w(0))}{\eps(p)(\eps(p) + 2U\rho w(0))}.$$
Similarly as before we will split the integration into the regions $|p| \le p_0$ and $|p| > p_0$, where $p_0$ is still the same as in \eqref{eps_bound}. For $|p| > p_0$ we have $\eps(p) > c$, hence
$$\frac{U^2\rho^2 w(0)(1-w(0))}{\eps(p)(\eps(p) + 2U\rho w(0))} \le \frac{U^2\rho^2 (1-w(0))}{c(c + 2U\rho w(0))} \le C \rho^2$$
for constant $C$ independent of $\rho$. Using this bound we get
$$\int_{|p| > p_0} \frac{U^2\rho^2 (1-w(0))}{\eps(p)(\eps(p) + 2U\rho w(0))} \le \int_{|p| > p_0} \frac{U^2\rho^2 (1-w(0))}{c(c + 2U\rho w(0))} \le\int_{\Lh} \frac{U^2\rho^2 (1-w(0))}{c(c + 2U\rho w(0))} \le C\rho^2.$$
For the integral with $|p| \le p_0$ we use analogous argument as in \eqref{small_p} to obtain
$$\int_{|p| \le p_0}\frac{U^2\rho^2 (1-w(0))}{\eps(p)(\eps(p) + 2U\rho w(0))} \le C U\rho^2 \frac{1}{\sqrt{2U\rho w(0)}} \arctan \left(\frac{p_0}{\sqrt{2U\rho w(0)}}\right)
\le C\rho^{3/2}.$$
Using this results in \eqref{sp_series} we conclude
$$\left|\frac{1}{|\Lambda_L|}\sum_{p \ne 0}\left(\frac{c_p}{1 + c_p} + \rho w(0)\right)\right| \le C\rho^{3/2}$$
and so
$$\frac{1}{|\Lambda_L|}\left(\sum_{p\ne 0} \left(\frac{c_p}{1 + c_p} + \rho  w(0)\right)\right)^2 = |\Lambda_L|\left(\frac{1}{|\Lambda_L|}\sum_{p\ne 0}\left(\frac{c_p}{1 + c_p} + \rho w(0)\right)\right)^2  \le C|\Lambda_L|\rho^3.$$
In the thermodynamic limit this gives the asymptotics
$$\frac{1}{|\Lambda_L|^2}\left(\sum_{p\ne 0} \big(s_p + \rho  w(0)\big)\right)^2  \lesssim \rho^3.$$
To prove \eqref{negligible_N} we perform a very similar argument as above: first we observe
$$\left|\sum_{p \ne 0}\frac{c_p}{1 - c_p}\right| = \left|\sum_{p \ne 0} \frac12\left(1 + \frac{2U\rho(1-w(0))}{\eps(p) + 2U\rho w(0)}\right)^{-1/2} - \frac12\right| \le \sum_{p \ne 0} \left[\frac12 - \frac12\left(1 + \frac{2U\rho(1-w(0))}{\eps(p) + 2U\rho w(0)}\right)^{-1/2}\right].$$
Using the inequality $(1 + x)^{-1/2} \ge 1 - \frac12x$ and approximating sum with the integral we get
$$\left|\frac{1}{|\Lambda_L|}\sum_{p \ne 0}\frac{c_p}{1 - c_p}\right| \lesssim \int_{\Lh} \frac{U\rho(1-w(0))}{\eps(p) + 2U\rho w(0)}dp.$$
The integral is convergent by similar arguments as before. Its value is independent of $L$ (it is dependent on $\rho$, but for this particular bound this fact is irrelevant), hence the proof of the lemma is finished.
\end{proof}

From Lemma \ref{min_prop} and Lemma \ref{lemma:remainder} we deduce the following corollary concerning the energy.

\begin{cor}
For the values of $c_p$ for which those minima of \eqref{min_target} are attained we have
\begin{equation}
\label{energy_final_error}
\begin{split}
\la H_L^\GC \ra_\Psi &= \sum_{p \ne 0} \frac12 \left(\sqrt{\left(\eps(p) + 2U\rho\right)\left(\eps(p) + 2U\rho w(0)\right)} - \eps(p) - U\rho\left(1 + w(0)\right)\right)
\\&+ \frac{U}{2}|\Lambda_L|\rho^2  - \frac{U}{2}|\Lambda_L|\rho^2w(0)^2 + |\Lambda_L| \cdot O(\rho^3)_{\rho \to 0} + |\Lambda_L|\cdot o(1)_{L \to \infty}.
\end{split}
\end{equation}
\end{cor}

\subsection{Thermodynamic limit}
Now we pass with the expression \eqref{energy_final_error} divided by the volume $|\Lambda_L|$ to the thermodynamic limit. We will denote this limit as
$$\lim_{L \to \infty} \frac{1}{|\Lambda_L|}\la H_L^\GC \ra_\Psi = e_\Psi$$
The sums in \eqref{energy_final_error} are Riemann sums of a continuous function on $\Lh$, hence they converge to the integrals of the proper expression. More precisely, we obtain
\begin{equation}
\label{energy_thermo}
\begin{split}
e_\Psi &= |\Lh|^{-1}\int_{\Lh}\frac12 \left(\sqrt{\left(\eps(p) + 2U\rho\right)\left(\eps(p) + 2U\rho w(0)\right)} - \eps(p) - U\rho\left(1 + w(0)\right)\right)dp
\\&+ \frac{U}{2}\rho^2  - \frac{U\rho^2}{2}w(0)^2 + O(\rho^3)
\end{split}
\end{equation}
To see the dependence on the scattering length $a$ we recall the definition \eqref{def:scattering_len} and write
$$U = 8\pi \a (1 + U\gamma),$$
hence
$$\frac{U}{2}\rho^2 = 4\pi \a \rho^2 + 4\pi \a U\gamma \rho^2,$$
Next note that (e.g. from \eqref{useful})
$$4\pi \a = \frac{w(0)}{2\gamma},$$
so we can rewrite the last line of \eqref{energy_thermo} (besides the error term) as
\begin{align*}
\frac{U}{2}\rho^2  - \frac{U\rho^2}{2}w(0)^2 &= 4\pi \a \rho^2 + 4\pi \a U\gamma \rho^2  - \frac{U\rho}{2}w(0)^2
\\&= 4\pi \a \rho^2 + \frac{U\rho^2}{2}w(0)\big(1 - w(0)\big) 
\\&= 4\pi \a \rho^2 + \frac{1}{2} U^2\gamma \big(1 - w(0)\big)^2\rho^2.
\end{align*}
Recalling also that $\gamma$ is given by the integral of the function $1/2\eps(p)$, this additionally allows to rewrite \eqref{energy_thermo} as
\begin{equation}
\label{energy_thermo2}
\begin{split}
e_\Psi &= |\Lh|^{-1}\int_{\Lh}\frac12 \left(\sqrt{\left(\eps(p) + 2U\rho\right)\left(\eps(p) + 2U\rho w(0)\right)}  - \eps(p) - U\rho\left(1 + w(0)\right) + \frac{U^2(1 - w(0))^2\rho^2}{2\eps(p)} \right)dp
\\&+ 4\pi \a\rho^2  + O(\rho^3),
\end{split}
\end{equation}
where we have joined the previous integral with the integral defining $\gamma$.

Let us now focus on evaluating the above integral. We note that the integrand is a positive function, which follows from the computation
\begin{equation}
\label{positive_integrand}
\begin{split}
&\;\sqrt{\left(\eps(p) + 2U\rho\right)\left(\eps(p) + 2U\rho w(0)\right)} - \eps(p) - U\rho\left(1 + w(0)\right) + \frac{U^2(1 - w(0))^2\rho^2}{2\eps(p)}
\\&= \frac{4U^2\rho^2w(0) - U^2\rho^2(1+w(0))^2}{\sqrt{\left(\eps(p) + 2U\rho\right)\left(\eps(p) + 2U\rho w(0)\right)} + \eps(p) + U\rho\left(1 + w(0)\right)} + \frac{U^2(1 - w(0))^2\rho^2}{2\eps(p)}
\\&= -\frac{U^2\rho^2(1-w(0))^2}{\sqrt{\left(\eps(p) + 2U\rho\right)\left(\eps(p) + 2U\rho w(0)\right)} + \eps(p) + U\rho\left(1 + w(0)\right)} + \frac{U^2(1 - w(0))^2\rho^2}{2\eps(p)}
\\&= \frac{U^2\rho^2(1-w(0))^2\left(\sqrt{\left(\eps(p) + 2U\rho\right)\left(\eps(p) + 2U\rho w(0)\right)} - \eps(p) + U\rho\left(1 + w(0)\right)\right)}{2\eps(p)\left(\sqrt{\left(\eps(p) + 2U\rho\right)\left(\eps(p) + 2U\rho w(0)\right)} + \eps(p) + U\rho\left(1 + w(0)\right)\right)}>0.
\end{split}
\end{equation}
Next, similarly as before, we will split the integration into two regions, but this time into regions $\eps(p) \ge \delta$ and $\eps(p) < \delta$ where $\delta = \delta(\rho)$ will be chosen as a certain function of $\rho$. For the first region if $\rho$ is sufficiently small (with respect to $U$ and $w(0)$) we can Taylor expand the square root up to the terms of order $\rho^3$. Using the inequality
$$(1+x)^{1/2} \le 1 + \frac12 x - \frac18x^2 + \frac{1}{16}x^3$$
we get
\begin{align*}
\sqrt{\left(\eps(p) + 2U\rho\right)\left(\eps(p) + 2U\rho w(0)\right)} &= \eps(p)\sqrt{\left(1 + \frac{2U\rho}{\eps(p)}\right)\left(1 + \frac{2U\rho w(0)}{\eps(p)}\right)}
\\&= \eps(p)\sqrt{1 + \frac{2U\rho(1+w(0))}{\eps(p)} + \frac{4U^2\rho^2w(0)}{\eps(p)^2}}
\\&\le \eps(p) + U\rho(1+w(0)) + \frac{2U^2\rho^2w(0) - \frac12U^2\rho^2(1+w(0))^2}{\eps(p)} + \frac{C}{\eps(p)^2}\rho^3
\\&\le \eps(p) + U\rho(1+w(0)) - \frac{\frac12U^2\rho^2(1-w(0))^2}{\eps(p)} + C(\delta)\rho^3,
\end{align*}
where $C > 0$ is a constant coming from the Taylor expansion, independent of $\rho$, and
\begin{equation}
\label{C_delta}
C(\delta) := \max_{\eps(p) \ge \delta} \frac{C}{\eps(p)^2} = \frac{C}{\delta^2},    
\end{equation}
which is another constant, dependent only on $\delta$. Now we can estimate the integrand as follows
\begin{align*}
&\quad \sqrt{\left(\eps(p) + 2U\rho\right)\left(\eps(p) + 2U\rho w(0)\right)} - \eps(p) -  U\rho\left(1 + w(0)\right) +  \frac{U^2(1 - w(0))^2\rho^2}{2\eps(p)}  
\\&\le \eps(p) + U\rho(1+w(0)) - \frac{\frac12U^2\rho^2(1-w(0))^2}{\eps(p)} + C(\delta)\rho^3 -\eps(p) - U\rho\left(1 + w(0)\right) + \frac{U^2(1 - w(0))^2\rho^2}{2\eps(p)}
\\&= C(\delta)\rho^3,
\end{align*}
hence
\begin{equation}
\label{integral>delta}
\begin{split}
&\;|\Lh|^{-1}\int_{\eps(p) \ge \delta}\frac12 \left(\sqrt{\left(\eps(p) + 2U\rho\right)\left(\eps(p) + 2U\rho w(0)\right)} - \eps(p) - U\rho\left(1 + w(0)\right) + \frac{U^2(1 - w(0))^2\rho^2}{2\eps(p)}\right)dp
\\&\le |\Lh|^{-1}\int_{\eps(p) \ge \delta}\frac12 C(\delta) \rho^3dp
\le |\Lh|^{-1}\int_{\Lh}\frac12 C(\delta) \rho^3dp
=\frac12 C(\delta)\rho^3.
\end{split}
\end{equation}
We will explicitly choose $\delta = \delta(\rho)$ after the next step.

Now we proceed to the integral on the domain $\eps(p) < \delta$. Using the coarea formula (see e.g. \cite[Section 3.4.3]{EvansMeasure}) for some general and sufficiently regular function $f$ integrable near zero we have
\begin{equation*}
\int_{\eps(p) < \delta} f(\eps(p))dp = \int_0^\delta f(r) \left(\int_{\eps(p) = r} \frac{d \mathcal H^2(\xi)}{|\nabla \eps(\xi)|}\right)dr,
\end{equation*}
where $\H^2$ is the two dimensional Hausdorff (surface) measure. By direct investigation of the formula \eqref{def:eps} and \eqref{eps_bound} we see that
$$c |p|^2 \le \eps(p) \le C|p|^2$$
for some constants $c$ and $C$ independent of $p$, hence
$$\H^2\big(\{p \colon \eps(p) = r\}\big) \le Cr$$
for some constant $C > 0$. Moreover
$$|\nabla \eps(p)| > c|p|$$
for some constant $c > 0$, hence for $\xi \in \{p \colon \eps(p) = r\}$ we have
$$\frac{1}{|\nabla \eps(\xi)|} \le \frac{1}{c|\xi|} \le \frac{C}{\sqrt{r}}$$
and therefore
$$\int_{\eps(p) = r} \frac{d \mathcal H^2(\xi)}{|\nabla \eps(\xi)|} \le Cr^{1/2}.$$
As a result, if the function $f$ is non-negative, we get
\begin{equation*}
\int_{\eps(p) < \delta} f(\eps(p))dp \le C \int_0^\delta r^{1/2} f(r) dr
\end{equation*}
for a constant $C > 0$ independent of $\delta$. We are going to use this observation for the integrand as in \eqref{energy_thermo2}, that is
$$f(\eps(p)) = \frac12 \left(\sqrt{\left(\eps(p) + 2U\rho\right)\left(\eps(p) + 2U\rho w(0)\right)} - \eps(p) - U\rho\left(1 + w(0)\right) + \frac{U^2(1 - w(0))^2\rho^2}{2\eps(p)}\right).$$
In \eqref{positive_integrand} we have already noted that this function is positive. We have
\begin{align*}
&\;|\Lh|^{-1}\int_{\eps(p) \le \delta}f(\eps(p))dp
\le C\int_0^{\delta} r^{1/2}f(r)dr
\\&= C\rho\int_0^{\delta} r^{1/2}\left(\sqrt{\left(\frac{r}{\rho} + 2U\right)\left(\frac{r}{\rho} + 2U w(0)\right)} - \frac{r}{\rho} - U\left(1 + w(0)\right) + \frac{U^2(1 - w(0))^2\rho}{2r}\right)dr
\\&= C\rho^{5/2}\int_0^{\delta/\rho} s^{1/2}\left(\sqrt{\left(s + 2U\right)\left(s + 2U w(0)\right)} - s - U(1+w(0)) + \frac{U^2(1-w(0))^2}{2s}\right)ds
\\& \le C\rho^{5/2}\int_0^{+\infty} s^{1/2}\left(\sqrt{\left(s + 2U\right)\left(s + 2U w(0)\right)} - s - U(1+w(0)) + \frac{U^2(1-w(0))^2}{2s}\right)ds,
\end{align*}
where in the second to last equality we have changed the variable $r := \rho s$ and in the last inequality we used the fact the integrand is well defined and positive on $\R_+$. Performing similar computation as in \eqref{positive_integrand} one can check that this integral is convergent, in particular we can conclude
$$|\Lh|^{-1}\int_{\eps(p) \le \delta}f(\eps(p))dp \le C\rho^{5/2}.$$
Now we combine this result with \eqref{integral>delta} and \eqref{C_delta} where we choose $\delta = \rho^{1/4}$. We eventually obtain 
$$|\Lh|^{-1}\int_{\Lh}f(\eps(p))dp \le C\rho^{5/2}.$$
This allows us to estimate $e_\Psi$ in \eqref{energy_thermo2} as
$$e_\Psi \le 4\pi \a \rho^2(1 + C\rho^{1/2}).$$
This finishes the proof of Proposition \ref{main_thm_gc}.

\subsection{Equivalence of ensembles}
\label{sec:Equiv}
We will now recall the well known argument which shows how to use Proposition \ref{main_thm_gc} in order to obtain the lower bound in Theorem \ref{main_thm}. We follow the proof from \cite[Lemma 3.3.2]{Aaen} with some inspiration from \cite[Lemma A.4]{BasCenSch}. Here, however, we will not assume that $\frac{N}{|\Lambda_L|}$ is constant. We start with the following lemma.
\begin{lemma}
\label{E(N,L)_lower}
For any $N$ and $L$ larger than the hopping length $R_0(t)$ of the lattice (recall the definition \eqref{def:hopping_length}) we have an inequality
$$\frac{E_0(N,L)}{|\Lambda_L|} \ge e_0 \left({\frac{N}{|\Lambda_L|}}\right).$$
\end{lemma}
\begin{proof}
For fixed $L$ and for any $k \in \N$  denote
$$L(k) := k(L+1) -1$$
and consider a (periodic) lattice $\Lambda_{L(k)}$. This lattice can be divided into $k^3$ sub-lattices with each sub-lattice being the translation of the original lattice $\Lambda_L$.\footnote{Note that the definition of $L(k)$ is correct as the number of points in $\Lambda_L$ is $|\Lambda_L| =(L+1)^3$, hence $|\Lambda_{L(k)}| = k^3(L+1)^3$.}

Next take any $N$-particle state $\psi_N \in \H^N_{L}$. Using this state for any $k \in \N$ we will construct a state on $\H^{k^3N}_{L(k)}$, i.e. the $k^3N$-particle Hilbert space based on a larger lattice $\Lambda_{L(k)}$. To this end on each of sub-lattices that $\Lambda_{L(k)}$ can be divided into we put a translated, independent copy of $\psi_N$ and define a state $\psi_{k^3N} \in L^2\left(\H^{k^3N}_{L(k)}\right)$ as the symmetrized result of this procedure. 

Now, since the interaction potential of the Bose-Hubbard model has zero range, different sub-lattices do not interact with each other, hence the total interaction potential energy is the sum of potential energies of each sub-lattice. Furthermore, due to the construction of the state, hopping between different sub-lattices gives the same contribution to the kinetic energy as hopping within a single sub-lattice with imposed periodic boundary condition (here we also use the condition $L \ge R_0(t)$), hence the kinetic energy of the state $\psi_{k^3N}$ is equal to the sum of kinetic energies of copies of $\psi_N$ from each sub-lattice. This allows us to conclude
$$E_0(k^3N, L(k)) \le \la H_{k^3 N,L(k)}\ra_{\psi_{k^3 N}} = k^3\la H_{L}\ra_{\psi_N}.$$
Minimizing over $\psi_N$ gives
$$E_0(k^3N, L(k)) \le k^3E_0(N,L),$$
and therefore
$$\frac{E_0(N,L)}{|\Lambda_L|} \ge \frac{E_0(k^3N, L(k))}{|\Lambda_{L(k)}|} = \frac{E_0(k^3N, L(k))}{k^3|\Lambda_{L}|}.$$
Since this inequality is valid for any $k$, we can pass to the limit $k \to \infty$ and obtain
$$\frac{E_0(N,L)}{|\Lambda_L|} \ge e_0\left(\frac{N}{|\Lambda_L|}\right)$$
as desired.
\end{proof}

Now we can proceed to the main problem. We first observe that trivially
$$E_0^{\GC}(N,L) \le E_0(N,L) $$
as any canonical trial state with $N$ particles can be lifted to the grand canonical one, occupying only the $N$-particle sector of the Fock space (in particular having $N$ as the expected number of particles). It follows that
$$\limsup_{\substack{N \to \infty \\ L \to \infty \\ N/|\Lambda_L| \to \rho}} \frac{E_0^{GC}(N,L)}{|\Lambda_L|}\le \limsup_{\substack{N \to \infty \\ L \to \infty \\ N/|\Lambda_L| \to \rho}} \frac{E_0(N,L)}{|\Lambda_L|} = e_0(\rho).$$
It remains to prove
$$\liminf_{\substack{N \to \infty \\ L \to \infty \\ N/|\Lambda_L| \to \rho}} \frac{E_0^{\GC}(N,L)}{|\Lambda_L|}\ge e_0(\rho).$$
To this end we introduce a variable $\mu \in \R$ (that can be interpreted as the chemical potential) and for any normalized $\Psi \in \F_L$ with $\la \NN\ra_{\Psi} = N$ we write
\begin{align*}
\frac{\la H\ra_\Psi}{|\Lambda_L|} &= \frac1{|\Lambda_L|} \left[\mu\la \NN\ra_{\Psi} + \la H - \mu \NN\ra_{\Psi}\right]
\\&= \frac1{|\Lambda_L|}\left[\mu N + \sum_{n=1}^\infty \|\Psi^{(n)}\|^2\left(\la H_n\ra_{\Psi^{(n)}} - \mu n\right)\right]
\\&= \mu \frac{N}{|\Lambda_L|} + \sum_{n=0}^\infty \|\Psi^{(n)}\|^2\left(\frac{\la H_n\ra_{\Psi^{(n)}}}{|\Lambda_L|} - \mu \frac{n}{|\Lambda_L|} \right)
\\&\ge \mu \frac{N}{|\Lambda_L|} + \sum_{n=0}^\infty \|\Psi^{(n)}\|^2\left(\frac{E_0(n,L)}{|\Lambda_L|} - \mu \frac{n}{|\Lambda_L|}\right)
\\&\ge \mu \frac{N}{|\Lambda_L|} + \sum_{n=0}^\infty \|\Psi^{(n)}\|^2\left(e_0\left(\frac{n}{|\Lambda_L|}\right) - \mu \frac{n}{|\Lambda_L|}\right)
\\&\ge \mu \frac{N}{|\Lambda_L|} + \sum_{n=0}^\infty \|\Psi^{(n)}\|^2 \inf_{\tilde \rho \ge 0}\left(e_0\left(\tilde \rho\right) - \mu \tilde \rho\right)
\\&= \mu\frac{N}{|\Lambda_L|} + \inf_{\tilde \rho \ge 0}\left(e_0\left(\tilde \rho\right) - \mu \tilde \rho\right),
\end{align*}
where in one of the steps we have used Lemma \ref{E(N,L)_lower}. We also recognize
$$\inf_{\tilde \rho \ge 0}\left(e_0\left(\tilde \rho\right) - \mu \tilde \rho\right) = -e_0^*(\mu),$$
where $e_0^*(\mu)$ is the Legendre transform of $e_0(\tilde \rho)$ (we also use notation $\tilde \rho$ in order not to confuse it with $\rho$ fixed in the statement of the Theorem \ref{main_thm}). Since $\Psi \in \F_L$ above was arbitrary, we conclude
$$\frac{E_0^{\GC}(N,L)}{|\Lambda_L|} \ge \mu\frac{N}{|\Lambda_L|} - e_0^*(\mu).$$
Taking the limes inferior of both sides gives
$$\liminf_{\substack{N \to \infty \\ L \to \infty \\ N/|\Lambda_L| \to \rho}} \frac{E_0^{\GC}(N,L)}{|\Lambda_L|} \ge \mu \rho - e_0^*(\mu).$$
Furthermore, as the left hand side is independent of $\mu$, we additionally get
$$\liminf_{\substack{N \to \infty \\ L \to \infty \\ N/|\Lambda_L| \to \rho}} \frac{E_0^{\GC}(N,L)}{|\Lambda_L|} \ge \sup_{\mu \in \R} \left(\mu \rho - e_0^*(\mu)\right) = e^{**}_0(\rho) = e_0(\rho),$$
where the last equality follows from the fact $e_0$ is a convex and continuous (up to a boundary) function and for such functions Legendre transform is an involution (i.e. $f^{**} = f$, see e.g. \cite[Lemma A.3]{BasCenSch} for a simple proof). The ends the proof of the upper bound in Theorem \ref{main_thm}.

\section{The lower bound}
\label{sec:Lower_bound}

In this section we will prove Proposition \ref{prop:lower_bound}. We will follow the strategy described at the beginning of the paper.

\subsection{Division into sub-lattices}
Similarly as in the proof of the lower bound in the continuous setting (see e.g. \cite[Chapter 2]{MathBEC}) we will divide the large (thermodynamic) lattice $\Lambda_L$ into smaller ones. The upcoming lemma is a well-known result, here will give a proof based on \cite[Lemma 5.21]{Rougerie} adapted to the lattice setting. Beforehand, in analogy to the definition \eqref{def:GSE}, we will denote the ground state energy of the $N$ particle system in the box of side-length $L$ with Neumann Laplacian as
$$E_0^\Neu(N,L) = \inf_{\substack{\psi \in \H^N_L\\\|\psi\| = 1}} \la\psi, H_{N,L}^\Neu\psi\ra$$
with
$$H_{N,L}^\Neu = -\sum_{i = 1}^N\Delta_{\Lambda_L,i}^\Neu + U\sum_{i<j}^N\delta_{x_i,x_j}$$
and $\Delta^\Neu$ defined (in accordance to the Appendix \ref{subsec:Neumann_bd}) as
$$-\Delta^\Neu_{\Lambda_L} u(x) = \sum_{\substack{y \sim x\\y \in \Lambda_L}} t(y-x)\big(u(x) - u(y)\big), \quad u \in L^2(\Lambda_L)$$
In the following proof we will use the fact the ground state of the Hamiltonian $H_{N,L}$ with either periodic or Neumann Laplacian over the symmetric wave functions is the same as the ground state over all wave functions, in particular the ground state energies are the same with or without imposing the symmetry constraint. This statement for the continuous case is proven e.g. in \cite[Corollary 3.1]{StabilityofMatter}, for discrete systems this proof is also valid with some straightforward modifications. 
\begin{lemma}
\label{box_division}
Choose $L$ and $\ell \in 2\N$ such that $\ell \ge R_0(t)$ and $\frac{L+1}{\ell+1} \in \N$. We have the following estimate of the ground state energy
$$E_0(N,L) \ge \inf \left\{\sum_{n=0}^N c_n E^{\Neu}_0(n,\ell) \colon c_n \ge 0, \; \sum_{n=0}^N c_n = \frac{(L+1)^3}{(\ell+1)^3}, \; \sum_{n=0}^Nn c_n = N\right\}.$$
\end{lemma}
\begin{proof}
Divide the (periodic) lattice $\Lambda_L =: \Lambda$ into $J$ smaller sub-lattices $\Lambda_j$, $j=1,\dots,J$ of side-length $\ell$, i.e. translated lattices $\Lambda_\ell$ (also note that $J = (\frac{L+1}{\ell+1})^3$). Next decompose the $N$-particle space $\Lambda^N_L$ with regard how many particles are in some box $\Lambda_j$, that is
\begin{equation}
\label{division_properties}
\begin{split}
\Lambda^N &= \bigcup_{\alpha} \Lambda^N_\alpha ,\\
\Lambda^N_\alpha &:= \Lambda_1^{\alpha_1} \times \dots \times \Lambda_J^{\alpha_J}
\end{split}
\end{equation}
where the union is taken with respect to all multi-indices $\alpha = (\alpha_1,\dots,\alpha_J)$ with $\alpha_j \in \N_0$ and $|\alpha| = \sum_{j=1}^J \alpha_J = N$, with the convention that if $\alpha_j = 0$ for some $j$ then $\Lambda_j$ is omitted in the cartesian product. The value $\alpha_j$ is the number of particles in box $\Lambda_j$. We also note that $\Lambda^N_\alpha \cap \Lambda^N_{\alpha'} = \emptyset$ for $\alpha \ne \alpha'$.

Denote by $\psi_N$ the ground state of the Hamiltonian $H_N$ on the (large) lattice $\Lambda_L$. For $x \in \Lambda^N$ and $y \in \Lambda$ we will denote $x[i \to y] = (x_1,\dots,y,\dots,x_N)$ where $y$ replaces $x_i$ on the $i$-th coordinate. We will also denote $V(x) = U\delta_{0,x}$, i.e. the on-site interaction potential. Then we have
\begin{equation}
\label{box_division_computation}
\begin{split}
&E_0(N,L) =  \sum_{i=1}^N \la \psi_N, - \Delta_i \psi_N\ra  + \frac{1}{2}\sum_{i,j=1}^N \la \psi_N, V(x_i-x_j) \psi_N\ra
\\&= \sum_{i=1}^N \frac12\sum_{x \in \Lambda^N} \sum_{y \sim x_i} t(x_i -y)|\psi_N(x[i \to y]) - \psi_N(x)|^2 + \frac{1}{2}\sum_{i,j=1}^N \sum_{x \in \Lambda^N}V(x_i - x_j)|\psi_N(x)|^2
\\&= \sum_{\alpha} \left(\sum_{i=1}^N \frac12\sum_{x \in \Lambda^N_\alpha} \sum_{y \sim x_i}t(x_i - y)|\psi_N(x[i \to y]) - \psi_N(x)|^2 + \frac{1}{2}\sum_{i,j=1}^N \sum_{x \in \Lambda^N_\alpha}V(x_i - x_j)|\psi_N(x)|^2\right)
\\&\ge \sum_{\alpha} \left(\sum_{i=1}^N \frac12\sum_{x \in \Lambda^N_\alpha} \sum_{\substack{y \sim x_i\\x[i \to y] \in \Lambda^N_\alpha}}t(x_i -y)|\psi_N(x[i \to y]) - \psi_N(x)|^2 + \frac{1}{2}\sum_{i,j=1}^N \sum_{x \in \Lambda^N_\alpha}V(x_i - x_j)|\psi_N(x)|^2\right).
\end{split}
\end{equation}
In the last inequality we simply neglected all the graph edges that connect the point $x_i$ to points lying outside the particular lattice in $\Lambda^N_\alpha$ in which point $x_i$ is included. For fixed $\alpha$ we recognize the expression under the sum as the evaluation in the state\footnote{Note that this is the moment where it is important that we do not consider only symmetric wave functions.} $\1_{\Lambda^N_\alpha}(x)\psi_N(x)$ of the expectation of $H_{N}\big|_{L^2(
\Lambda^N_\alpha)}$, i.e the $N$-body Hamiltonian restricted to the space $L^2(\Lambda^N_\alpha)$ with Laplacian being the Neumann Laplacian $\Delta^\Neu$ (see Appendix \ref{subsec:Neumann_bd}). Due to the translation invariance of the system and the fact that the different sub-lattices in $\Lambda^N_\alpha$ do not interact between one another (there is no hopping due to the Neumann Laplacian and there is no interaction due to the fact the interaction has zero range) we have a bound on the ground state energy of this system
$$ \inf_{\|\psi\|=1}\left\la \psi, H_{N}\big|_{L^2(
\Lambda^N_\alpha)}\psi\right\ra \ge \sum_{j = 1}^J E_0^\Neu(\alpha_j, \ell)$$
hence we can further bound \eqref{box_division_computation} as
$$E_0(N,L) \ge \min_\alpha \inf_{\|\psi\|=1} \left\la H_{N}\big|_{L^2(
\Lambda^N_\alpha)} \right\ra_{\psi} \sum_\alpha\|\1_{\Lambda^N_\alpha}\psi_N\| = \min_\alpha \inf_{\|\psi\|=1} \left\la H_{N}\big|_{L^2(
\Lambda^N_\alpha)} \right\ra_{\psi} \ge \min_\alpha \sum_{j = 1}^J E_0^\Neu(\alpha_j, \ell)$$
where the equality in the middle follows from \eqref{division_properties} and the fact $\psi_N$ is normalized. Next we can regroup the terms in the sum with respect to the number $c_n$ ($n = 1,\dots,N$) of boxes with exactly $n$  particles inside, which gives:
$$E_0(N,L) \ge \min\left\{\sum_{n=0}^N c_n E_0^{\Neu}(n,\ell) \colon c_n \in \N_0, \quad \sum_{n=0}^N c_n = J, \quad \sum_{n=0}^Nnc_n = N\right\},$$
where the first constraint means that there need to be exactly $J = (\frac{L+1}{\ell+1})^3$ boxes, the second constraint means the total number of particles needs to be $N$. For now coefficients $c_n$ needed to be integers, however we can extend the minimum to the infimum over real positive $c_n$'s satisfying given constraints. This extension may only lower the infimum and finishes the proof.
\end{proof}

Sometimes it is useful to express the inequality from the Lemma by the coefficients of relative number of boxes, i.e. we replace $c_n$ with $(\frac{\ell+1}{L+1})^3c_n$. Under this replacement we get the following corollary.
\begin{cor}
With the same assumption as in Lemma \ref{box_division}
\label{box_relative}
$$E_0(N,L) \ge \frac{(L+1)^3}{(\ell+1)^3}\inf \left\{\sum_{n=0}^N c_n E^{\Neu}_0(n,\ell) \colon 0 \le c_n \le 1, \; \sum_{n=0}^N c_n = 1, \; \sum_{n=0}^Nn c_n = \frac{N(\ell+1)^3}{(L+1)^3}\right\}.$$    
\end{cor}

\subsection{Neumann eigenvalues}

In order to use the above corollary successfully we need to understand the spectrum of the Neumann Laplacian $-\Delta^\Neu_{\Lambda_L}$ on the lattice $\Lambda_L$. In full generality it won't be possible to derive explicit formula for the eigenvalues, hence we will need to estimate them in a proper way. Before doing that we will derive the explicit form of the spectrum for a very special case of the neighborhood relation on the lattice $\Lambda$.
\begin{lemma}
\label{lemma:special_Neumann}
Assume that $t(v) = 0$ for $v \not \in D_1$ (recall the notation from \eqref{D_1_def}), i.e. the $x \sim y$ if and only if $x-y$ or $y-x$ is a primitive vector of the lattice $\Lambda$. Then the eigenvalues of the Neumann Laplacian $-\Delta^\Neu_{\Lambda_L}$ are given by
\begin{equation}
\label{Neumann_eigenvalue_special}
\epsN_\sp(k) = \sum_{i=1}^3t(a_i)\Bigg(2 - 2\cos \left(\frac{k_i\pi}{L+1}\right) \Bigg) = 4\sum_{i=1}^3t(a_i) \sin^2 \left(\frac{k_i\pi}{2(L+1)}\right),
\end{equation}
where $k = (k_1,k_2,k_3) \in \{0,1,\dots,L\}^3$.
\end{lemma}

\begin{proof}
At first we will consider the Neumann Laplacian on the one-dimensional interval $\Omega := \left[0, L\right] \cap \Z$. Our goal is to find functions $u \in L^2(\Omega)$ such that
\begin{equation}
\label{Neumann_problem}
-\Delta_\Omega^{\Neu} u = \lambda u
\end{equation}
for some $\lambda \in \R$ (real as this operator is self-adjoint). In fact, as noted in Remark \ref{Laplacian_rem}, to solve this problem it is enough to consider the equation
$$-\Delta u(x) = \lambda u(x) \text{ for } x = 0,1,\dots,L$$
on the extended lattice $[-1,L+1] \cap \Z$ with additional (boundary) conditions
\begin{equation}
\label{Neumann_cond_1d}
u(-1) = u(0), \quad u(L) = u(L+1),
\end{equation}
where $\Delta$ is the standard lattice Laplacian defined in \eqref{Laplacian_def}. This problem has a well-known solution. One obtains a family $\{u_k\}_{k=0,\dots,L}$ of $L+1$ functions satisfying Neumann eigenvalue problem \eqref{Neumann_problem}
$$u_k(x) = A_k \cos \frac{k\pi(x+\frac12)}{L+1}, \quad A_k \text{ - normalization constant}$$
corresponding to eigenvalues
\begin{equation}
\label{Neuman_eigenvalue_1D}
\lambda_k = 2\left(1 - \cos \frac{k\pi}{L+1}\right) = 4 \sin^2 \left(\frac{k\pi}{2(L+1)}\right).
\end{equation}
The normalization constants of $u_k$ can be computed explicitly: for $k = 0$ we have $A_0 = (L+1)^{-1/2}$ and for $k \ne 0$ we have $A_k = \left(\frac{2}{L+1}\right)^{1/2}$.

Due to the fact that the Neumann Laplacian is self-adjoint, functions $u_k$ are orthogonal to each other as each of them corresponds to a different eigenvalue. Moreover, as their number is equal to the dimension of the space $L^2(\Omega)$, this system is in fact an orthogonal basis. 

Going back to the interval $[-\frac{L}{2},\frac{L}{2}] \cap \Z$, due to the translation invariance, shifted functions
$$w_k := u_k\left(\cdot + \frac{L}{2}\right)$$
are the orthogonal eigenfunctions of the Neumann Laplacian on this domain. In particular eigenvalues \eqref{Neuman_eigenvalue_1D} remain unchanged.

Now we return to the problem of finding eigenvalues of the Neumann Laplacian on $\Lambda_L$. Motivated by the previous results we make an ansatz: for $x = m_1 a_1 + m_2a_2 + m_3a_3 \in \Lambda_L$ and for $k = (k_1, k_2, k_3)$, $k_j = 0,1,\dots, L$ we define
$$\psi_k(x)= w_{k_1}(m_1)w_{k_2}(m_2)w_{k_3}(m_3).$$
Using the assumption that the hopping is allowed only in the directions of the primitive vectors of lattice $\Lambda$ function $\psi_k$ satisfies the Neumann boundary condition on the nearest neighbor closure $(\Lambda_L)_\nn$ (in the sense of equations \eqref{Neumann_cond_1d}), hence to compute $-\Delta^{\Neu}\psi_k(x)$ for $x \in \Lambda_L$ we can use the standard graph Laplacian $-\Delta$ as in \eqref{Laplacian_def} in the interior of the extended lattice $(\Lambda_L)_{\nn}$ (here we once again refer to Remark \ref{Laplacian_rem}). Moreover, once again using the assumption on weights, the graphs $\Lambda_L$ and  $(\Lambda_L)_{\nn}$ have a structure of graph cartesian product of graphs coming from one-dimensional discrete intervals, hence $L^2(\Lambda_L)$ is isomorphic to the tensor product of $L^2$ spaces on those intervals. This in particular implies that the system of tensor products of the basis vectors of one-dimensional spaces, which is exactly the system $(\psi_k)_{k}$, forms an orthogonal basis of $L^2(\Lambda_L)$.
 
To compute the eigenvalues explicitly we note from the part concerning the one dimensional intervals that the functions $w_{k_i}$ satisfy
\begin{equation}
\label{1dim_cor}
w_{k_i}(m+1) + w_{k_i}(m-1) = 2\cos \left(\frac{k_i\pi}{L+1}\right)w_{k}(m),
\end{equation}
therefore we have
\begin{equation}
\label{direction_sum}
\begin{split}
-\Delta \psi_k(x) &=  \left(\sum_{i=1}^3 2t(a_i) \left(1 - \cos \left(\frac{k_i\pi}{L+1}\right)\right) \right)\psi_k(x)
\\&= \left(4\sum_{i=1}^3 t(a_i)\sin^2 \left(\frac{k_i \pi}{2(L+1)}\right)\right) \psi_k(x):= \eps^\Neu_\sp(k) \psi_k(x).
\end{split}
\end{equation}
This ends the proof.
\end{proof}

Observe the similarity of the expression $\epsN_\sp(k)$ to the dispersion relation \eqref{def:eps} or its version \eqref{def:eps_finite} for a finite lattice. Recalling that $\pi \delta_{i,j} = \frac12 a_i \cdot b_j$ we can write
\begin{equation}
\label{Neu_per_relation}
\epsN_\sp(k) = \eps\left(k_1\frac{b_1}{2(L+1)} + k_2\frac{b_2}{2(L+1)} + k_3\frac{b_3}{2(L+1)}\right) = \eps\left(\frac{1}{2(L+1)}Bk\right).
\end{equation}

For a more general neighborhood relation in $\Lambda$ the approach used in the proof of Lemma \ref{lemma:special_Neumann} will be unsuccessful as in general the graph $\Lambda_L$ with edges corresponding to the Neumann Laplacian will not have a structure of a graph product of one-dimensional graphs. However we can still show that in the general case the Neumann eigenvalues are in some sense close to the periodic eigenvalues and derive some bounds. The following Lemma will give the precise statement of those observations and will be used in the next section to prove an appropriate bound on the Hamiltonian.
\begin{lemma}
\label{Laplace_comparison}
For fixed $L \in 2\N$ denote $-\Delta^\Neu$ as the Neumann Laplacian on $\Lambda_L$, $-\Delta^\Per$ as the periodic Laplacian and $-\Delta^\Neu_\sp$ as the Neumann Laplacian considered in Lemma \ref{lemma:special_Neumann}, that is with hopping only in the direction of the primitive translation vectors. Then the following statements hold true:
\begin{enumerate}[a)]
    \item $-\Delta^\Neu_\sp \le -\Delta^\Neu \le -\Delta^\Per$ in the sense of operator inequalities on $L^2(\Lambda_L)$.
    \item The spectral gap $\epsG(L)$ of $-\Delta^\Neu$ (i.e. the difference between the lowest, here zero, eigenvalue and the second lowest eigenvalue) can be bounded as
    \begin{equation}
    \label{epsG_asymptotics}
    \frac{c_\gap}{(L+1)^2} \le \epsG(L) \le \frac{C_\gap}{(L+1)^2},   
    \end{equation}
    where $c_\gap$ and $C_\gap$ are some constants independent of $L$.
    \item Denote by $P_+$ the projection onto $\{\chi_0\}^\perp$, i.e the space orthogonal to the space spanned by the constant function in $L^2(\Lambda_L)$. Then for sufficiently large $L$
    $$\frac{1}{|\Lambda_L|}|\Tr \left[(-P_+\Delta^\Neu P_+)^{-1}\right] - \Tr\left[(-P_+\Delta^\Per P_+)^{-1}\right]| \le C L^{-1/3}$$
    for some constant $C$ independent of $L$. The inverses of the operators are considered on the subspace $P_+L^2(\Lambda_L) = \{\chi_0\}^\perp$.
    
\end{enumerate}
\end{lemma}
\begin{proof} To prove  \textit{a)}, we recall that if $u \in L^2(\Lambda_L)$ then
$$\la u, -\Delta^\Neu u\ra = Q^\Neu(u) = \frac12\sum_{x \in \Lambda_L} \sum_{\substack{y \sim x\\y \in \Lambda_L}} t(y-x)|u(x) - u(y)|^2,$$
therefore
\begin{equation*}
\begin{split}
\frac12\sum_{x \in \Lambda_L} \sum_{\substack{y \in \Lambda_L\\y - x \in \pm D_1}} t(y-x)|u(x) - u(y)|^2 = \la u, -\Delta^\Neu_\sp u\ra \le \la u, -\Delta^\Neu u\ra
\end{split}
\end{equation*}
and
\begin{equation*}
\la \psi, -\Delta^\Neu \psi\ra \le \la \psi ,-\Delta^\Per \psi\ra = \frac12\sum_{x \in \Lambda_L} \sum_{\substack{y \in \Lambda_L\\(y - x)_\Per \in \pm D}} t(y-x)|\psi(x) - \psi(y)|^2,
\end{equation*}
where $(y-x)_\Per$ is the difference $y-x$ interpreted as an element of the group $(A\Z^3)/((L+1)A\Z^3)$, i.e. hopping through the boundary is included as wrapping to the other side of the lattice. Those inequalities prove the first point of the lemma.

To prove \textit{b)} we note that by the previous point and the min-max principle (see e.g \cite[Theorem 4.2.6]{Horn} we have
$$\lambda_2(-\Delta^\Neu_\sp) \le \lambda_2(-\Delta^\Neu) \le \lambda_2(-\Delta^{\Per}),$$
i.e. the spectral gap of $-\Delta^\Neu$ is bounded below by the spectral gap of $-\Delta^\Neu_\sp$ and bounded above by the spectral gap of $-\Delta_\Per$. Using explicit formulas \eqref{Neumann_eigenvalue_special} and \eqref{def:eps_finite} for eigenvalues of those operators as well as a similar argument that was used to prove inequality \eqref{eps_bound} the desired bounds follows easily.

In order to prove \textit{c)} we first recall Cavalieri's principle: for measurable space $(\Omega,\mu)$ and non-negative function $f$ on $\Omega$ we have
$$\int_\Omega f(x) d\mu(x) = \int_0^\infty \mu(\{x \in \Omega\colon f(x) > s\})ds.$$
Here we will denote by $\lambda_j(T)$, $j=1,2\dots$, the eigenvalues of $T$ arranged in a non-decreasing order and use this principle for the eigenvalue counting measure of some positive-definite matrix $T$ and function $f(x) = 1/x$, that is
$$\Tr T^{-1} = \int_0^\infty \#\left\{j \colon \frac1{\lambda_j(T)} > s\right\}ds = \int_0^\infty \frac{N_T(s)}{s^2}ds,$$
where
$$N_T(t) := \#\{j \colon \lambda_j(T)< t\} = \max \{j \colon \lambda_j(T) < t\}$$
is the spectral function counting the eigenvalues (with multiplicity and the convention that $\max \emptyset = 0)$. Denoting by $N_\Neu(t)$ and $N_\Per(t)$ the spectral functions of operators $(-P_+\Delta^\Neu P_+)$ and $(-P_+\Delta^\Per P_+)$ respectively we obtain
\begin{equation}
\label{Cavalieri}
\begin{split}
&|\Tr \left[(-P_+\Delta^\Neu P_+)^{-1}\right] - \Tr\left[(-P_+\Delta^\Per P_+)^{-1}\right]| \le \int_0^\infty \frac{|N_\Neu(s) -N_\Per(s)|}{s^2}ds \\&= \int_0^\delta \frac{|N_\Neu(s) -N_\Per(s)|}{s^2}ds + \int_\delta^\infty \frac{|N_\Neu(s) -N_\Per(s)|}{s^2}ds,   
\end{split}
\end{equation}
where $\delta = \delta(L)$ will be chosen later. We will bound each of those integrals separately.

To bound the second integral we will use the following fact \cite[Corollary 4.3.5]{Horn}: if $A$ and $B$ are Hermitian $n \times n$ matrices with $\rank (A - B) \le r$ then
$$\lambda_j(B) \le \lambda_{j+r}(A) \text{ for } j = 1,\dots, n-r$$
and
$$\lambda_j(B) \ge \lambda_{j-r}(A) \text{ for } j = r+1 ,\dots, n.$$
From this fact we can deduce the following bound on the difference of spectral functions of $A$ and $B$: for every $s > 0$
$$|N_A(s) - N_B(s)| \le r.$$
To see it we denote $N_A(s) = k$. If $k \le r$ then trivially
$$N_A(s) - N_B(s) \le r,$$
whereas if $k > r$ then
$$s\geq \lambda_k(A) \ge \lambda_{k-r}(B) \Rightarrow N_B(s) \ge k - r = N_A(s) - r.$$
To obtain the second inequality we reverse the roles of $A$ and $B$.

Applying this fact to $N_\Neu(s)$ and $N_\Per(s)$ we get
$$|N_\Neu(s) -N_\Per(s)| \le \rank (\Delta^\Neu - \Delta^\Per) \le CL^2$$
for some constant $C$ independent of $L$ as the difference $\Delta^\Neu - \Delta^\Per$ acts non-trivially only on the boundary $\partial \Lambda_L$ of $\Lambda_L$ and its nearest neighbors, which is the set of cardinality of order $L^2$. It follows that
$$\int_\delta^\infty \frac{|N_\Neu(s) -N_\Per(s)|}{s^2}ds \le \frac{CL^2}{\delta}.$$
To bound the first integral in \eqref{Cavalieri} we note that by point $(a)$ we have
$$|N_\Neu(s) - N_\Per(s)| \le N_\Neu(s) - N_\Per(s) \le N_\sp(s),$$
where $N_\sp(s)$ is the spectral function of $-P_+\Delta^\Neu_\sp P_+$. Moreover by point $(b)$ we have
$$N_\sp(s) = 0 \text{ for } s < \frac{c_\gap}{L^2},$$
hence we can only consider $s$ satisfying
\begin{equation}
\label{s_gap}
s \ge \frac{c_\gap}{L^2},
\end{equation}
which also implicitly imposes a condition $\delta \ge \frac{c_\gap}{L^2}$. Recall the explicit formula for the eigenvalues $\eps^\Neu_\sp(k)$ of the operator $-\Delta^\Neu_\sp$ given in \eqref{direction_sum}:
$$\eps^\Neu_\sp(k) = \left(4 \sum_{i=1}^3 t(a_i)\sin^2 \frac{p_i}{2}\right) \Bigg|_{p = \pi k/(L+1)}, \; k \in \{0,1\dots,L\}^3.$$
Let  $|\cdot|_\infty$ be the supremum norm of a vector in $\R^3$. By equivalence of norms in $\R^3$ and the same argument as the one used to prove inequality \eqref{eps_bound} there exist a constant $c > 0$ such that
$$4 \sum_{i=1}^3 t(a_i)\sin^2 \frac{p_i}{2} \ge c|p|^2_{\infty},$$
for all $p \in \Lh$. We further observe that
\begin{align*}
N_\sp(s) &= \# \left\{k \in \{0,1\dots,L\}^3\setminus\{(0,0,0)\} \colon 4 \sum_{i=1}^3 t(a_i)\sin^2 \frac{k_i\pi}{2(L+1)} < s\right\}
\\&\le \# \left\{k \in \{0,1\dots,L\}^3 \colon \frac{c\pi^2}{(L+1)^2}|k|^2_\infty < s\right\}    
\\&= \# \left\{k \in \{0,1\dots,L\}^3 \colon |k|_\infty < \frac{(L+1)}{\pi}\sqrt{\frac{s}{c}}\right\}
\\&\le C(Ls^{1/2} + 1)^3 = C s^{3/2}\left(L + \frac{1}{s^{1/2}}\right)^3 \le CL^3 s^{3/2},
\end{align*}
where the last inequality follows from \eqref{s_gap}. As a final result we get
$$\int_0^\delta \frac{|N_\Neu(s) - N_\Per(s)|}{s^2}ds \le \int_{c_\gap/L^2}^\delta \frac{N_\sp(s)}{s^2} \le CL^3 \int_{c_\gap/L^2}^\delta s^{-1/2}ds \le CL^3\delta^{1/2}.$$
Combining this with the previous estimate we eventually obtain
$$\frac{1}{|\Lambda_L|}|\Tr \left[(-P_+\Delta^\Neu P_+)^{-1}\right] - \Tr\left[(-P_+\Delta^\Per P_+)^{-1}\right]| \le C\left(\delta^{1/2} + \frac{1}{L\delta}\right),$$
which after optimizing in $\delta$ yields $\delta \sim L^{-2/3}$ and
$$\frac{1}{|\Lambda_L|}|\Tr \left[(-P_+\Delta^\Neu P_+)^{-1}\right] - \Tr\left[(-P_+\Delta^\Per P_+)^{-1}\right]| \le CL^{-1/3}.$$
This bound is valid for sufficiently large $L$ when $\delta \ge \frac{c_\gap}{L^2}$.
\end{proof}

The method used in the proof of point $(c)$ above can be used to obtain the following useful corollary.
\begin{cor}
\label{trace_power}
Within the setting like in Lemma \ref{Laplace_comparison} for any power $\nu > \frac32$ and for sufficiently large $L$ we have the bound
$$\frac{1}{|\Lambda_L|}\Tr (-P_+\Delta^\Neu P_+)^\nu \le C L^{2 - 3/\nu}$$
for some constant $C$ independent of $L$.
\end{cor}
\begin{proof}
Mimicking the previous proof we get
\begin{align*}
\Tr (-P_+\Delta^\Neu P_+)^\nu &= \int_{c_\gap/L^2}^\delta \frac{N_\sp(s^{1/\nu})}{s^2}ds + \int_{\delta}^\infty \frac{N_\sp(s^{1/\nu})}{s^2}ds
\\&\le CL^3 \int_{c_\gap/L^2}^\delta \frac{s^{3/2\nu}}{s^2}ds + \int_{\delta}^\infty \frac{(L+1)^3}{s^2}ds
\\&\le CL^3 \cdot (L^{-2})^{-1 + 3/2\nu} + CL^3
\\&\le C L^3 \cdot L^{2 - 3/\nu}.
\end{align*}
\end{proof}

\subsection{Analysis on sub-lattices}

In this section we will work with the eigenbasis of the Neumann Laplacian $-\Delta^\Neu_{\Lambda_\ell}$ for the box of size $\ell$, denoted as $\{\psi_k\}_{k \in \{0,1\dots,\ell\}^3}$. Such basis exists as this is a self-adjoint operator on a finite dimensional space. Moreover it is easy to check that the constant function, here denoted as $\psi_0$, is an eigenvector with eigenvalue zero. Furthermore, as the hopping constants defining $\Delta^\Neu_{\Lambda_L}$ are real, one has $\overline{\Delta^\Neu_{\Lambda_\ell} u(x)} = \Delta^\Neu_{\Lambda_\ell} \overline{u(x)}$ (in other words this operator is a complexification of a symmetric operator over the real vector space $L^2_\R(\Lambda_\ell))$) and therefore the eigenfunctions $\psi_k$ can be chosen as real-valued. Our goal is to prove the following proposition.
\begin{prop}
\label{H_N_lower_bound}
Assume that\begin{equation}
\label{N/L_bound}
\frac{n}{\ell+1} < \frac{c_\gap}{48\pi \a},   
\end{equation}
where $c_\gap$ is defined below \eqref{epsG_asymptotics}. Then the following operator inequality holds:
$$H_{n,\ell} \ge 4\pi \a \frac{n^2}{|\Lambda_\ell|} - C\left(\frac{n^2}{\ell^4}\log \ell\right) - C\left(\frac{n}{\ell^3}\right),$$
for some constant $C$ independent of $n$ and $\ell$. In particular
$$E_0^\Neu(n,\ell) \ge 4\pi \a \frac{n^2}{|\Lambda_\ell|} - C\left(\frac{n^2}{\ell^{10/3}}\right) - C\left(\frac{n}{\ell^3}\right).$$
\end{prop}

\begin{proof} The idea of the proof below is based on \cite[Section 1]{OptimalRate} and \cite[Lemma 5]{ChongLiangNam}.

\noindent \textit{Step 1.} We will express the Hamiltonian $H_{n,\ell}$ in terms of creation and annihilation operators in the Neumann basis introduced above: for $k = (k_1,k_2,k_3)$, $k_j = 0,\dots,L$ denote $a_k = a(\psi_k)$ and $a_k^* = a^*(\psi_k)$. Next denote by $P$ the projection onto the constant function $\psi_0 = \frac{1}{|\Lambda_\ell|^{1/2}} \in L^2(\Lambda_\ell)$ and by $m_\varphi$ the multiplication operator by a function $\varphi(x-y)$ (i.e. the scattering equation solution) acting on the two body Hilbert space $\H_{\ell}^{\otimes 2}$. We first note that we have an operator inequality
$$(\1 - P \otimes P m_\varphi) U \delta_{x,y} (\1 - m_\varphi P \otimes P) \ge 0.$$
This is equivalent to
$$U \delta_{x,y} \ge U m_\varphi \delta_{x,y} (P \otimes P) + U (P \otimes P) m_\varphi \delta_{x,y}- (P \otimes P) m_\varphi^2\delta_{x,y} (P \otimes P).$$
Now we will take the second quantization of both sides of this inequality and expressing it in the Neumann basis representation. As the system of $\psi_k(x)$ is the orthonormal basis of the one-body space we have
\begin{align*}
\la \psi_p \otimes \psi_q, & (U m_\varphi \delta_{x,y} \; P \otimes P) \psi_k \otimes \psi_r\ra = \la \psi_p \otimes \psi_q, (U m_\varphi \delta_{x,y}) \psi_0 \otimes \psi_0\ra \delta_{k,0}\delta_{r,0}
\\&= \delta_{k,0}\delta_{r,0}\cdot \frac{U}{|\Lambda_\ell|}\sum_{x, y \in \Lambda_L}\psi_p(x)\psi_q(y)\varphi(x-y) \delta_{x,y}
 = \delta_{k,0}\delta_{r,0}\delta_{q,p} \cdot \frac{U\varphi(0)}{|\Lambda_\ell|}.
\end{align*}
Here we have used the fact $\psi_k$ can be chosen to be real-valued, so that the complex conjugation in the inner product can be omitted. A similar computation for other matrix elements leads to
$$\la \psi_p \otimes \psi_q, (P \otimes P \;U m_\varphi \delta_{x,y}) \psi_k \otimes \psi_r\ra = \delta_{p,0}\delta_{q,0}\delta_{r,k} \cdot \frac{U\varphi(0)}{|\Lambda_\ell|}$$
and
$$\la \psi_p \otimes \psi_q, (P \otimes P \;U m_\varphi^2 \delta_{x,y} \; P \otimes P) \psi_k \otimes \psi_r\ra = \delta_{p,0}\delta_{q,0}\delta_{k,0}\delta_{r,0} \cdot \frac{U\varphi(0)^2}{|\Lambda_\ell|}.$$
As the end result we obtain
$$H_{n,\ell} \ge \sum_{k \ne 0} \epsN(p)a_k^*a_k + \frac{U \varphi(0)}{2|\Lambda_\ell|}\sum_{k \ne 0}\left(a^*_ka^*_ka_0a_0 + a^*_0a^*_0a_ka_k\right) + \frac{U(2\varphi(0) -\varphi(0)^2)}{2|\Lambda_\ell|}a^*_0a^*_0a_0a_0.$$

\noindent \textit{Step 2.} We will use the following operator inequality (see \cite[Theorem 6.3]{LiebSolovej}):
$$A(b_k^*b_k + b^*_{-k}b_{-k}) + B(b^*_kb^*_{-k} + b_kb_{-k}) \ge -(A - \sqrt{A^2 - B^2})\frac{[b_k,b_k^*] + [b_{-k},b^*_{-k}]}{2},$$
valid for any operators $b_k$, $b_{-k}$, $b^*_k$ and $b^*_{-k}$ on the Fock space satisfying $[b_k,b_{-k}] = [b^*_k,b^*_{-k}] = 0$ and any constants $0 \le B \le A$. Here we will use it for
$$b_k = n^{-1/2}a_0^*a_k, \quad b_k^* = n^{-1/2}a_k^*a_0, \quad k \ne 0$$
and $b_{-k} := b_k$. One can check that 
$$b_k^*b_k \le a_k^*a_k, \quad [b_k, b^*_k] \le \1, \quad a^*_ka^*_ka_0a_0 = nb^*_kb^*_k, \quad  a^*_0a^*_{0}a_ka_k = nb_kb_k.$$
Now we note that the assumption on $\frac{n}{\ell+1}$ and the lower bound in \eqref{epsG_asymptotics} imply that 
$$24 \pi \a \frac{n}{|\Lambda_\ell|} = 24 \pi \a \frac{n}{(\ell+1)^3} < \frac12 \frac{c_\gap}{(\ell+1)^2} < \frac12 \epsG(\ell),$$
hence the (open) interval
$$\left(16 \pi \a \frac{n}{|\Lambda_\ell|}, \frac12 \epsG(\ell) - 8\pi \a\frac{n}{|\Lambda_\ell|}\right)$$
is nonempty. This allows to choose parameter $\mu$ satisfying
\begin{equation}
\label{mu_assumption}
16 \pi \a \frac{n}{|\Lambda_\ell|} < \mu <  \frac12\epsG(\ell) - 8 \pi \a \frac{n}{|\Lambda_\ell|}.
\end{equation}
In the end we obtain (with denoting $n_0 := a_0^* a_0$):
\begin{align*}
H_{n,\ell} &= \sum_{k \ne 0} (\epsN(k) - \mu)a^*_ka_k + \frac{U \varphi(0)}{2|\Lambda_\ell|}\sum_{k \ne 0}\left(a^*_ka^*_ka_0a_0 + a^*_0a^*_0a_ka_k\right) + \frac{U(2\varphi(0) -\varphi(0)^2)}{2|\Lambda_\ell|}a^*_0a^*_0a_0a_0 + \mu \NN_+
\\&\ge \frac12\sum_{k \ne 0}\left[ (\epsN(k) - \mu) (b^*_kb_k + b^*_{k}b_k) + \frac{nU\varphi(0)}{|\Lambda_\ell|}(b^*_kb^*_k + b_kb_k)\right] + \frac{U(2\varphi(0) -\varphi(0)^2)}{2|\Lambda_\ell|}n_0(n_0-1) + \mu \NN_+
\\&\ge - \frac{1}{2}\sum_{k \ne 0} \left[\epsN(k) - \mu - \sqrt{(\epsN(k) - \mu)^2 - \frac{n^2U^2\varphi(0)^2}{|\Lambda_\ell|^2}}\right] + \frac{U(2\varphi(0) -\varphi(0)^2)}{2|\Lambda_\ell|}n_0(n_0-1) + \mu \NN_+.
\end{align*}
Recalling that
$$8 \pi \a  = U \varphi(0) = \frac{U}{1 + U\gamma}$$
we see that the square root is well defined as
$$\epsN(k) - \mu - \frac{nU\varphi(0)}{|\Lambda_\ell|} = \epsN(k) - \mu - 8\pi \a \frac{n}{|\Lambda_\ell|} > \frac12 \epsG(\ell) > 0,$$
in which we have used the upper bound on $\mu$ from \eqref{mu_assumption}.

\noindent \textit{Step 3.} Using the inequality
$$1 - \sqrt{1 - x} \le \frac{1}{2}x + \frac18 x^2$$
 we get
\begin{align*}
\epsN(k) - \mu - \sqrt{(\epsN(k) - \mu)^2 - \frac{n^2U^2\varphi(0)^2}{|\Lambda_\ell|^2}} \le \frac{n^2U^2\varphi(0)^2}{2|\Lambda_\ell|^2(\epsN(k)-\mu)} + \frac18 \left(\frac{n^4U^4\varphi(0)^4}{|\Lambda_\ell|^4(\epsN(k) - \mu)^3}\right).
\end{align*}
By the condition $\mu < \frac12\epsG(\ell)$ and the lower bound in \eqref{epsG_asymptotics} we get
\begin{equation*}
\begin{split}
\frac{1}{\epsN(k) - \mu} & =\frac{1}{\epsN(k)} + \frac{\mu}{\epsN(k)^2 \left(1 - \frac{\mu}{\epsN(k)}\right)}
\\&\le \frac{1}{\epsN(k)} + \frac{\frac12 \epsG(\ell)}{\epsN(k)^2\left(1 - \frac12\frac{\epsG(\ell)}{\epsN(k)}\right)} \le \frac{1}{\epsN(k)} + \frac{C}{\ell^2}\frac{1}{\epsN(k)^2}
\end{split} 
\end{equation*}
for some constant $C > 0$ independent of $\ell$. Next, using $U^2\varphi(0)^2 \le \frac{1}{\gamma}<C$, we get:
\begin{align*}
\sum_{k \ne 0}\frac{n^2U^2\varphi(0)^2}{4|\Lambda_\ell|^2(\epsN(k) - \mu)} &\le \sum_{k \ne 0}\frac{n^2U^2\varphi(0)^2}{4|\Lambda_\ell|^2\epsN(k)} + \frac{C}{\ell^2}\sum_{k \ne 0}\frac{n^2}{|\Lambda_\ell|^2\epsN(k)^2}
\\&\le \sum_{k \ne 0}\frac{n^2U^2\varphi(0)^2}{4|\Lambda_\ell|^2\epsN(k)} + C\frac{n^2}{\ell^{9/2}} ,
\end{align*}
since by Corollary \ref{trace_power}
$$\frac{1}{(\ell+1)^3}\sum_{k \ne 0} \frac{1}{\epsN(k)^2} \le C\ell^{1/2}.$$
Moreover, using $\frac{n}{\ell+1} \le C$ and an inequality
$$\frac{1}{\epsN(k) - \mu} \le \frac{C}{\epsN(k)}$$
that can be justified by similar arguments as above we obtain
$$\sum_{k \ne 0}\frac18 \left(\frac{n^4U^4\varphi(0)^4}{|\Lambda_\ell|^4(\epsN(k) - \mu)^3}\right) \le C \frac{n^2}{(\ell+1)^{10}} \sum_{k \ne 0} \frac{1}{\epsN(k)^3} \le C\frac{n^2}{\ell^6}$$
as (once again from Corollary \ref{trace_power})
$$\frac{1}{(\ell+1)^3}\sum_{k \ne 0} \frac{1}{\epsN(k)^3} \le C\ell.$$
From the above it follows that
$$\frac{1}{2}\sum_{k \ne 0} \left[\epsN(k) - \mu - \sqrt{(\epsN(k) - \mu)^2 - \frac{n^2U^2\varphi(0)^2}{|\Lambda_\ell|^2}}\right] \le \sum_{k \ne 0} \frac{n^2U^2\varphi(0)^2}{4|\Lambda_\ell|^2\epsN(k)} + C\frac{n^2}{\ell^{9/2}}.$$
Now, using the first point of Lemma \ref{Laplace_comparison} and the min-max principle, we have
$$\sum_{k \ne 0} \frac{n^2U^2\varphi(0)^2}{4|\Lambda_\ell|^2\epsN(k)} = \sum_{k \ne 0} \frac{n^2U^2\varphi(0)^2}{4|\Lambda_\ell|^2\eps(k)} + O\left(\frac{n^2}{\ell^{10/3}}\right).$$
By observation \eqref{Neu_per_relation}, the above sum is a Riemann sum for the integral of the function $k \mapsto \eps(Bk)$ on the domain $k \in [0,1/2]^3$ hence (by the standard Riemann sum approximation argument)
\begin{equation}
\label{Riemann_approx}
\sum_{k \ne 0}\frac{n^2U^2\varphi(0)^2}{2|\Lambda_\ell|^2\epsN(k)} = \frac{n^2U^2\varphi(0)^2}{2\left(\frac{1}{2}\right)^3|\Lambda_\ell|}\int_{[0,\frac{1}{2}]^3}\frac{dk}{\eps(Bk)} + O\left(\frac{n^2}{\ell^4}\log \ell\right).
\end{equation}
Using the symmetry $k_j \leftrightarrow (-k_j)$ of the function under the integral and then changing variables $p = Bk$ we get
$$\int_{[0,\frac{1}{2}]^3}\frac{dk}{\eps(Bk)} = \frac{1}{8}\int_{[-\frac{1}{2},\frac{1}{2}]^3}\frac{dk}{\eps(Bk)} = \frac{1}{8|\det B|}\int_{B[-\frac12,\frac12]^3}\frac{dp}{\eps(p)} = \frac{1}{8|\Lh|}\int_{\Lh}\frac{dp}{\eps(p)}.$$
As a result
$$\eqref{Riemann_approx} = \frac{n^2U^2\varphi(0)^2}{2|\Lambda_\ell| |\Lh|}\int_{\Lh}\frac{dp}{\eps(p)} + O\left(\frac{n^2}{\ell^4}\log \ell\right) = \frac{n^2U^2\varphi(0)^2\gamma}{|\Lambda_\ell|} + O\left(\frac{n^2}{\ell^4}\log \ell\right).$$
We also note that all of the obtained previously error terms decay faster than $\frac{n^2}{\ell^{10/3}}$, therefore in the next step we will include all of them in the $O\left(\frac{n^2}{\ell^{10/3}}\right)$ term.

\noindent \textit{Step 4.} Gathering all of the estimates we conclude
\begin{align*}
H_{n,\ell} &\ge \mu \NN_+ - \frac12\frac{n^2U^2\varphi(0)^2\gamma}{|\Lambda_\ell|}  + \frac{U(2\varphi(0) -\varphi(0)^2)}{2|\Lambda_\ell|}(n - \NN_+)(n - \NN_+-1) - C\left(\frac{n^2}{\ell^{10/3}}\right)
\\&= \mu \NN_+ + \frac{n^2}{2|\Lambda_\ell|}\left(-\frac{U^2\gamma}{(1 + U\gamma)^2} + \frac{U + 2U^2\gamma}{(1 + U\gamma)^2}\right)
\\&\;+  \frac{U + 2U^2\gamma}{2|\Lambda_\ell|(1 + U\gamma)^2} \left(-2n \NN_+ + \NN_+^2 - n + \NN_+\right) - C\left(\frac{n^2}{\ell^{10/3}}\right)
\\&\ge \mu \NN_+ + \frac{n^2}{2|\Lambda_\ell|} \frac{U}{1 + U\gamma} - \frac{2Un \NN_+}{|\Lambda_\ell|(1 + U\gamma)} - C\left(\frac{n^2}{\ell^{10/3}}\right) - C\left(\frac{n}{\ell^3}\right)
\\&= (\mu - 16\frac{n}{|\Lambda_\ell|} \pi \a) \NN_+  + 4\pi \a \frac{n^2}{|\Lambda_\ell|} - C\left(\frac{n^2}{\ell^{10/3}}\right) - C\left(\frac{n}{\ell^3}\right)
&\\&\ge 4\pi \a \frac{n^2}{|\Lambda_\ell|} - C\left(\frac{n^2}{\ell^{10/3}}\right) - C\left(\frac{n}{\ell^3}\right),
\end{align*}
where in the last inequality we used the lower bound from \eqref{mu_assumption} and non-negativity of $\NN_+$. This ends the proof.
\end{proof}

\begin{remark}
In the proof above we did not have to use the Neumann symmetrization technique used in \cite{ChongLiangNam} as in the discrete setting the Neumann Laplacian \eqref{Neumann_Laplacian_def} is defined for all functions on $\Lambda_\ell$, in particular for the restriction $\varphi|_{\Lambda_\ell}$ of the scattering equation solution. We refer to the Remark \ref{Laplacian_rem} in the Appendix for more discussion concerning this fact.
\end{remark}

\subsection{Conclusion}

Now we can conclude the proof of Proposition \ref{prop:lower_bound}, hence finishing the proof of the main Theorem \ref{main_thm}. The following proof is based on \cite[Corollary 1.3]{BoccatoSeiringer} (which itself is similar to the proof given in \cite[Chapter 2.2]{MathBEC}).

\begin{proof}[Proof of Proposition \ref{prop:lower_bound}]
For fixed $\rho > 0$ we define
\begin{equation}
\label{l_GP_def}
\ell = \ell(\rho) = \left\lceil\left(\frac{192 \pi \a}{c_\gap} \rho\right)^{-1/2}\right\rceil - 1,
\end{equation}
where $c_\gap$ is once again the constant from \eqref{epsG_asymptotics}. As the thermodynamic limit does not depend on the choice of sequences $N \to \infty$, $L \to \infty$ with $N/|\Lambda_L| \to \rho$ we will  consider only the values of $L$ such that $\frac{L+1}{\ell+1}$ is an integer. This will allow us to use the localization method. We will also assume that the sequence $\frac{N}{|\Lambda_L|}$ tends to $\rho$ from below, i.e $\frac{N}{|\Lambda_L|} \le \rho$ for every considered $N$ and $L$. This is a purely technical assumption, related to the fact that we are dealing with only discrete values of $N$ and $L$, hence we cannot assume that $\frac{N}{|\Lambda_L|} = \rho$ for every $N$ and $L$, as this would significantly restrict the possible values of $\rho$.

As in the proof of Lemma \ref{box_division} we split the thermodynamic lattice (of side length $L$) into sub-lattices of side length $\ell$ and introduce parameter $p$ defined as
\begin{equation}
\label{p_def}
p := \frac{c_\gap(\ell + 1)}{48\pi \a}.
\end{equation}
By Proposition \ref{H_N_lower_bound}, for $n$ satisfying $n < p$ we have
$$E_0^\Neu(n,\ell) \ge 4\pi \a \left(\frac{n^2}{|\Lambda_\ell|} - C\frac{n^2}{\ell^{10/3}} - C\frac{n}{\ell^3}\right)$$
For $n \ge p$ we use the fact the interaction potential is non-negative (in particular $H_{n,\ell}$ is a non-negative operator), so that the ground state energy is super-additive:
$$E_0^\Neu(n_1 + n_2,\ell) \ge E_0^\Neu(n_1,\ell) + E_0^\Neu(n_2,\ell).$$
With this fact for $n \ge p$ we have
$$E_0^\Neu(n,\ell) \ge \left\lfloor \frac{n}{p}\right\rfloor E_0^\Neu(p,\ell) \ge \frac{n}{2p}E_0^\Neu(p,\ell).$$
Using Corollary \eqref{box_relative} we obtain
\begin{equation}
\label{box_bound}
E_0(N,L) \ge \frac{4 \pi \a |\Lambda_L|}{|\Lambda_\ell|^2} \inf \left\{\sum_{n < p} c_n (n^2 - C\frac{n^2}{\ell^{1/3}} - Cn) + \frac12\sum_{n \ge p} c_n n\left(p - C\frac{p}{\ell^{1/3}} - C\right)\right\},  
\end{equation}
where the infimum is taken with the constraints stated in the Corollary. Defining
$$\xi:= 1 - \frac{C}{\ell^{1/3}}$$
we rephrase the minimization problem to finding minimum of
$$\sum_{n < p} c_n (\xi n^2 - Cn) + \frac12\sum_{n \ge p} c_n n\left(\xi p - C\right).$$
To this end we additionally define
$$r: = \sum_{n < p} c_n n.$$
We note that $r \le \frac{N(\ell+1)^3}{(L+1)^3}$ and, by convexity of the function
$$x \mapsto F(x): = \xi x^2 - Cx, \quad F(0) = 0$$
we have
$$\sum_{n < p} c_n (\xi n^2 - Cn) = \sum_{n < p} c_nF(n) + \sum_{n \ge p} c_n F(0) \ge F\left(\sum_{n < p} c_nn\right) = F(r) = \xi r^2 - Cr.$$
As a result
\begin{equation}
\label{quadratic_estimate}
\sum_{n < p} c_n (n^2 - Cn) + \frac12\sum_{n \geq p} c_n n\left(\xi p - C\right) \ge \xi r^2 - Cr + \frac12\left(\frac{N(\ell+1)^3}{(L+1)^3} - r\right)(\xi p-C).    
\end{equation}
When considering $r \in \R$ the above quadratic function attains its minimum at $r = \frac{\xi p + C}{4\xi}$. Restricting the domain to $r \in [0,\frac{N(\ell+1)^3}{(L+1)^3}]$ with our choice of $p$ the minium is attained at $r = \frac{N(\ell+1)^3}{(L+1)^3}$ as
$$\frac{N(\ell+1)^3}{(L+1)^3} \le \frac14 p,$$
which can be seen from equivalent (by \eqref{l_GP_def} and \eqref{p_def}) inequality
$$\rho \ge \frac{N}{|\Lambda_L|} \cdot  \frac{192 \pi \a}{c_\gap}\left(\left\lceil\frac{192 \pi \a}{c_\gap}\right\rceil\right)^{-1},$$
which is true by the technical assumption made at the beginning. The minimal value of the right hand side of \eqref{quadratic_estimate} on this interval is
$$\xi \left(\frac{N(\ell+1)^3}{(L+1)^3}\right)^2 - C\frac{N(\ell+1)^3}{(L+1)^3} = \xi \left(\frac{N|\Lambda_\ell|}{|\Lambda_L|}\right)^2 - C\frac{N|\Lambda_\ell|}{|\Lambda_L|}.$$
Recalling the inequality \eqref{box_bound}, the definition of $\xi$ and the definition \eqref{l_GP_def}, this implies
$$E_0(N,L) \ge \frac{4 \pi \a N^2}{|\Lambda_L|}\left(1 - \frac{C}{\ell^{1/3}} - C\frac{|\Lambda_L|}{N|\Lambda_\ell|}\right) \ge \frac{4 \pi \a N^2}{|\Lambda_L|}\left(1 - C\rho^{1/6} - C\frac{|\Lambda_L|}{N}\rho^{3/2}\right).$$
In the thermodynamic limit we get
\begin{equation*}
\begin{split}
e_0(\rho) &= \lim_{\substack{N \to \infty \\ L \to \infty \\ N/L^3 \to \rho}}\frac{E_0(N,L)}{|\Lambda_L|} \ge \lim_{\substack{N \to \infty \\ L \to \infty \\ N/L^3 \to \rho}} \frac{4 \pi a N^2}{|\Lambda_L|^2}\left(1 - C\rho^{1/6} - C\frac{|\Lambda_L|}{N}\rho^{3/2}\right)
\\&\ge 4\pi \a \rho^2\left(1 - C\rho^{1/6} - C\rho^{1/2}\right)
\\& \ge 4\pi \a \rho^2\left(1 - C\rho^{1/6}\right)
\end{split}
\end{equation*}
for $\rho$ small enough. This ends the proof of Proposition \ref{prop:lower_bound}.
\end{proof}

\appendix

\section{Fourier analysis on a lattice}
\label{sec:Fourier}

\subsection{General theory}
In this subsection we will briefly recall the abstract setting in which the Fourier transform can be defined. We refer to e.g. \cite{Rudin_groups} for more systematic approach.

Let $G$ be locally compact abelian group, let $\mu$ be Haar measure on $G$ (that is $\mu(gA) = \mu(A)$ for any measurable subset $A \subset G$ and any element $g \in G$). We define the (Pontryagin) dual group $\widehat G$ as
$$\widehat G = \text{Hom}(G, S^1),$$
where $S^1$ is a unit circle. The group $\widehat G$ is also abelian and locally compact, moreover group $\widehat {\widehat G}$ is isomorphic to $G$.

For a function $f \in L^1(G,\mu)$ we define its Fourier transform $\widehat f:\widehat G \to \C$ as
$$\widehat f(\chi) = \int_G f(g)\overline{\chi(g)}d\mu(g).$$
As the Haar measure on $G$ is fixed, once can choose the normalization of the Haar measure $\nu$ on $\widehat G$ such that the following inversion formula (defined for a certain class of functions $f$) holds
$$f(x) = \int_{\widehat G }\widehat f(\chi) \chi(x) d \nu(\chi).$$
The Fourier transform also extends to the unitary operator from $L^2(G,\mu)$ to $L^2(\widehat G, \nu)$, which is sometimes referred to as the (generalization of) Plancherel theorem.

\subsection{Infinite Bravais lattice}

We will present the theory of the $d$-dimensional Bravais lattices in the context of the abstract Fourier analysis on groups. Similarly as in \eqref{def:Bravais_lattice} we define Bravais lattice $\Lambda$ as
$$\Lambda = A\Z^d = \left\{\sum_{j=1}^d m_j a_j \; \colon \; m_j \in \Z\right\},$$
where $A$ is a $d \times d$ invertible real matrix and vectors $a_j$ are its columns. In this context vectors $a_j$ are called the primitive (translation) vectors of the lattice $\Lambda$. We treat $\Lambda$ as an additive topological group with discrete topology and the standard counting measure as its Haar measure.

The reciprocal lattice $\Lambda^*$ is then defined as
$$\Lambda^* = \left\{y \in \R^d \colon y \cdot x \in 2\pi \Z \text{ for all } x \in \Lambda\right\}.$$
One can check that $\Lambda^*$ is also a Bravais lattice generated by primitive vectors $b_j$, $j=1,\dots,d$, defined by the relation
$$a_i \cdot b_j = 2\pi\delta_{i,j}$$
or equivalently
$$b_j = \frac{1}{2\pi}(A^T)^{-1}e_j, \quad e_j \text{ - standard basis vector in } \R^d.$$

The Pontryagin dual group $\Lh$ in this context is called the Brillouin zone of the lattice $\Lambda$. Since $\Lambda$ is finitely generated every homomorphism $\chi \in \text{Hom}(\Lambda,S^1)$ is uniquely determined by values
$$\chi(a_i) := e^{i\theta_j}$$
for some $\theta_j \in \R$, hence every such $\chi$ is of the form
$$\chi = \chi_\theta, \quad \theta \in \R^d$$
with
\begin{equation}
\label{chi_theta_def}
\chi_\theta\left(\sum_{j=1}^dm_ja_j\right) = e^{i\sum_j m_j\theta_j}.   
\end{equation}
We note that if $y \in \Lambda^*$ then $\chi_\theta \equiv \chi_{\theta + y}$ as for any $x = \sum_{j}m_ja_j\in \Lambda$ we have
$$\chi_{\theta+y}(x) = e^{i\sum_j m_j(\theta_j + y_j)} = \chi_\theta(x) e^{ix \cdot y} = \chi_\theta(x)$$
as $x \cdot y \in 2\pi\Z$. It follows that the Brillouin zone $\Lh$ can be isomorphically identified as a quotient group
$$\Lh = \R^d/ \Lambda^*.$$
We will further identify $\Lh$ as
\begin{equation}
\label{Lh_def}
\Lh = B\T^d = \left\{\sum_{j=1}^d b_jt_j \; \colon \; t_j \in \left[-\frac12,\frac12\right)\right\} \text{ with periodic boundary condition}.   
\end{equation}
Here $B$ is a matrix whose columns are the primitive vectors of the reciprocal lattice $b_j$ and $\T^d$ is the $d$-dimensional unit torus $\left[-\frac12,\frac12\right)$. This identification allows to identify the correct (in the sense of the Fourier inversion formula) Haar measure on $\Lh$ as the normalized Lebesgue measure on the set $B\T^d$. With this fact we can write the Fourier transform formula on $\Lambda$ and the inverse transform formula on $\Lh$ explicitly: for $f \in L^1(\Lambda)$ and $g \in L^1(\Lh)$ we have
\begin{equation}
\label{lattice_Fourier_transform}
\widehat f(p) = \sum_{x \in \Lambda}f(x)e^{-ip\cdot x}, \quad p \in \Lh
\end{equation}
and
\begin{equation}
\label{inverse_transform}
\check g(x) = |B\T^3|^{-1}\int_{B\T^3} g(p)e^{i p\cdot x}dp, \quad x \in \Lambda.
\end{equation}
By Plancherel theorem $f \mapsto \widehat f$ and $g \mapsto \check g$ extend to unitary maps between $L^2(\Lambda)$ and $L^2(\Lh)$ and the extensions are each others inverses. We also note that $|B\T^3| = |\det B|$.

\subsection{Distributional Fourier transform}

As we have already identified $\Lh \simeq \R^d/\Lambda^*$, we can use the fact that the latter has a structure of a compact manifold (diffeomorphic to the torus $\T^d$) and extend the definition of the Fourier transform on $\Lambda$ to a distributional one.

Let $\psi \colon \Lambda \to \C$ be a function with at most polynomial growth, that is
$$|\psi(x)| \le C(1 + |x|)^s$$
for some constants $C > 0$ and $s \in \R$. We define the distributional Fourier transform of $\psi$, also denoted by $\widehat \psi$, as a distribution on $\Lh$ given by
\begin{equation}
\label{Fourier_def}
\langle \widehat \psi, f \rangle = \sum_{x \in \Lambda} \overline{\psi(x)} \check{f}(x),
\end{equation}
where $f \in \mathcal{D}(\Lh) = C^{\infty}(\Lh)$ and $\check{f}$ denotes the inverse Fourier transform \eqref{inverse_transform}. Note that if $\psi \in L^1(\Lambda)$ then this definition coincides with the standard one \eqref{lattice_Fourier_transform}.

\subsection{Finite Bravais lattice}
\label{subsec:finite_lattice}
For a Bravais lattice $\Lambda$ as before and for $L \in 2\N$ we define a finite Bravais lattice $\Lambda_L$ as
\begin{equation*}
\begin{split}
\Lambda_L &= (A\Z^d)/(LA\Z^d)
\\&= \left\{\sum_{j=1}^d m_ja_j \colon \; m_j = -\frac{L}{2}, -\frac{L}{2} + 1, \dots, \frac{L}{2}, \; j=1,\dots,d\right\} \text{ with periodic boundary condition.} 
\end{split}
\end{equation*}
Once again this is an additive group with discrete topology.

Using similar arguments as for the infinite lattice we can show that every $\chi \in \text{Hom}(\Lambda_L,S^1)$ is of the form
\begin{equation}
\label{momentum_representation}
\chi(x) = \chi_p(x) = \frac{1}{|\Lambda_L|^{1/2}}e^{i p \cdot x}
\end{equation}
where $|\Lambda_L| = (L+1)^3$ is the number of points in $\Lambda_L$ and $p$ is the element of
\begin{equation}
\label{LhL_def}
\widehat \Lambda_L := \left\{\sum_{j=1}^d m_j \frac{b_j}{L+1} \colon m_j=-\frac{L}{2},-\frac{L}{2} + 1,\dots, \frac{L}{2} - 1, \frac{L}{2}\right\},
\end{equation}
where $b_j$ are primitive vectors of the reciprocal lattice $\Lambda^*$. We will use this identification of $\Lh_L$ for the entire paper.

Not that in \eqref{momentum_representation} we have introduced an additional normalization factor. The reason for it is that when we consider a standard counting measure on $\Lambda_L$ as its Haar measure then the system $\{\chi_p\}_{p \in \Lh_L}$ forms an orthonormal basis of $L^2(\Lambda_L)$. We will refer to this system as the momentum basis of $L^2(\Lambda_L)$.

The correct choice for the Haar measure on $\Lh_L$ is again the standard counting measure on $\Lh_L$. Once again we can write the formulae for the Fourier transform on $\Lambda_L$ and its inverse on $\Lh_L$ explicitly as
\begin{equation}
\label{Fourier_transform_finite}
\widehat f(p) = \frac{1}{|\Lambda_L|^{1/2}}\sum_{x \in \Lambda_L} f(x) e^{-i p \cdot x} = \la \chi_p, f \ra_{L^2(\Lambda_L)}, \quad p \in \Lh_L
\end{equation}
and
\begin{equation}
\check g(x) = \frac{1}{|\Lambda_L|^{1/2}} \sum_{p \in \Lh_L} g(p) e^{ip \cdot x} = \la\overline \chi_p, g\ra_{L^2(\Lh_L)} \quad x \in \Lambda_L.
\end{equation}
As the Fourier transform is unitary we also note the Parseval identity
$$\sum_{x \in \Lambda_L}\overline{f(x)}g(x) = \sum_{p \in \Lh_L} \overline{\widehat f(p)}\widehat g(p).$$

\section{Graph calculus}

\subsection{Basic definitions}
It will be useful to work within the graph calculus formalism. A graph $G$ is a couple $G = (V,E)$, where $V$ is a finite\footnote{We can also consider infinite, but countable sets of vertices. This however requires adding some technical assumptions on summability of functions on vertices and edges.} set of vertices and $E \subset V \times V$ is the set of edges. Note that the edges are directed, that is $(x,y) \ne (y,x)$ for $x \ne y$. This approach will allow to define directional derivative. We will consider only non-oriented graphs, which in this setting means
$$(x,y) \in E \Rightarrow (y,x) \in E.$$
We will also assume that there are no self-loops (that is there are no edges of the form $(x,x)$). We will say that $x$ and $y$ are nearest neighbors if $(x,y) \in E$. This defines a symmetric relation on $V$ that will be denoted as $x \sim y$. Moreover to each edge $(x,y) \in E$ we will assign a positive real number $t(x,y)$, which gives rise to the weighted graph structure. Here we will also assume that $t(x,y) = t(y,x)$ for every edge $(x,y)$.

Consider a subset $\Omega \subset V$. We define the boundary of $\Omega$, denoted $\partial \Omega$, as
$$\partial \Omega = \{x \in \Omega \colon \text{there exists } y \not \in \Omega, y\sim x\}$$
We also define the set of interior edges $E_\Omega$ of the set $\Omega$ as
$$E_\Omega = \{(x,y) \in E \colon x,y \in \Omega\}.$$
It will be useful to also define the "nearest neighbors boundary" of the set $\Omega$ defined as
$$\partial_\nn \Omega := \{y \not \in \Omega \colon y \sim x \text{ for some }x \in \partial \Omega\}$$
and the "nearest neighbors closure" of $\Omega$
$$\Omega_{\nn} := \Omega \cup \partial_\nn\Omega.$$

With a graph we can associate two Hilbert spaces: the space of functions on the vertices $L^2(V)$ and functions on the edges $L^2(E)$ (both with counting measure). For a function $f: V \to \C$ we define its (discrete) gradient $\nabla f: E \to \C$ as
$$\nabla f(x,y) = \sqrt{t(x,y)}\left(f(y) - f(x)\right).$$
For a given edge $(x,y) \in E$ the value $\nabla f (x,y))$ may be considered as the directional derivative in direction $x \to y$. As $\nabla : L^2(V) \to L^2(E)$ we can consider its dual $\nabla^*:L^2(E) \to L^2(V)$ which satisfies the property that for any $f \in L^2(V)$ and $F \in L^2(E)$ we have
$$\la F, \nabla f\ra_{L^2(E)} = \la \nabla^*F, f\ra_{L^2(V)}.$$
We can also define $\nabla^*$ explicitly by the formula
\begin{equation}
\label{nabla_star_def}
\nabla^* F (x) = \sum_{y \sim x} \sqrt{t(x,y)}\left(F(y,x) - F(x,y)\right).
\end{equation}
Next we can define the discrete divergence $\dv : L^2(E) \to L^2(V)$ as $\dv =-\frac12 \nabla^*$ and the discrete Laplacian $\Delta: L^2(V) \to L^2(V)$ as $\Delta = \dv \circ \nabla$. We can check that the action of $\Delta$ can be written explicitly
\begin{equation}
\label{Laplacian_def}
\Delta f (x)= \sum_{y \sim x} t(x,y)(f(y) - f(x)) = \sum_{y \sim x} \nabla f(x,y).
\end{equation}
From the definition it is easy to see that the Laplace operator is self-adjoint on $L^2(V)$. We will show it for completeness: for $f,g \in L^2(V)$ we have
$$\la f, \Delta g\ra_{L^2(\Omega)} = \la g, -\frac12\nabla^*\nabla f\ra_{L^2(E)} = -\frac12 \la \nabla g, \nabla f\ra_{L^2(E)} = \la -\frac12 \nabla^* \nabla g, f\ra_{L^2(V)} = \la \Delta g, f\ra_{L^2(V)}$$

We are interested in deriving some properties of the discrete Laplace operator resembling the Green identities that hold for the standard (continuous) Laplacian. Let $\Omega \subset V$ be a fixed subset. Then we have
\begin{equation}
\label{discrete_by_parts}
\begin{split}
\sum_{x \in \Omega}\overline{g(x)}\Delta f(x) &= \la \1_{\Omega}g, \Delta f\ra_{L^2(V)} = -\frac12\la \1_\Omega g, \nabla^*\nabla f\ra_{L^2(V)} = -\frac12 \la \nabla \1_\Omega g, \nabla f\ra_{L^2(E)}
\\& = -\frac12\sum_{(x,y) \in E}\nabla (\1_\Omega \overline{g})(x,y) \cdot \nabla f(x,y)
\\& = -\frac12\sum_{(x,y) \in E_\Omega} \overline{\nabla g(x,y)} \nabla f(x,y) + \sum_{x \in \partial \Omega}\sum_{\substack{y \not \in \Omega\\ y \sim x}}\overline{g(x)}\nabla f(x,y).
\end{split}
\end{equation}
There is no factor $\frac12$ in the second term as cases $x \in \Omega$, $y \not \in \Omega$ and $x \not \in \Omega$, $y \in \Omega$ are symmetric and give the same contribution. The factor $\frac12$ in the first term is the side effect of considering the ordered pairs in the definition of the edge $(x,y)$, which essentially means every bond between $x$ and $y$ is counted twice. This result is the discrete analogue to the standard integration by parts formula
$$\int_{\Omega}\overline{g(x)}\Delta f(x) dx = - \int_{\Omega} \overline{\nabla g(x)}\nabla f(x) dx + \int_{\partial \Omega} \overline{g(x)} \frac{\partial f}{\partial n}(x) d\sigma(x).$$

\subsection{Neumann Laplacian}
\label{subsec:Neumann_bd}
We will define the Neumann Laplacian on some set $\Omega \subset V$. To this end we first define a quadratic form
\begin{equation}
\label{Neumann_form}
Q^{\Neu}(f) = \frac12\sum_{(x,y) \in E_\Omega} |\nabla f(x,y)|^2.   
\end{equation}
We note that the value of $Q^{\Neu}(f)$ depends only on the restriction of $f$ to the set $\Omega$. The Neumann Laplacian $-\Delta_{\Omega}^{\Neu}$ is defined as the operator on $L^2(\Omega)$ associated with this quadratic form, meaning that for every $f \in L^2(\Omega)$ there holds
$$\la f, -\Delta_\Omega^{\Neu}f\ra_{L^2(\Omega)} = Q(f).$$
We can write the action of $-\Delta_\Omega^{\Neu}$ explicitly: for $f \in L^2(\Omega)$
we have
\begin{equation}
\label{Neumann_Laplacian_def}
-\Delta_\Omega^{\Neu} f(x) = \sum_{\substack{y \in \Omega\\y \sim x}} t(x,y)(f(x) - f(y))  = \sum_{\substack{y \in \Omega\\y \sim x}} \nabla f(y,x).
\end{equation}

\begin{remark}
\label{Laplacian_rem}
Note that if the point $x$ is in the interior (i.e. not on the boundary) of $\Omega$ then the action of $-\Delta_\Omega^{\Neu}$ coincides with the action of the standard discrete Laplacian.  If $x \in \partial \Omega$ then the action of $-\Delta_\Omega^{\Neu}$ \textbf{looks as if} the function $f$ satisfied an additional condition
\begin{equation}
\label{Neumann_boundary}
\forall_{x \in \partial \Omega} \forall_{\substack{y \not \in \Omega\\y \sim x}} \; f(y) = f(x).   
\end{equation}
This can be interpreted as a discrete version of the standard Neumann condition $\frac{\partial f}{\partial n} = 0$ on $\partial \Omega$. We emphasize however that here the function $f$ needs to be defined only on the set $\Omega$ and not on the set of its nearest neighbors. Moreover, in some cases, imposing condition \eqref{Neumann_boundary} might be impossible -- a simple example of such situation is $\Omega = V \setminus \{v_0\}$ for some $v_0 \in V$, i.e. the set of all but one vertices. Then for a function $f \in L^2(\Omega)$ it is possible to impose \eqref{Neumann_boundary} if and only if the value of $f$ on all neighbors of $v_0$ is the same. This example illustrates the fact that the Neumann Laplacian is \textbf{not} the same as the standard Laplacian restricted to the functions satisfying Neumann boundary condition \eqref{Neumann_boundary}. However, if some function $f$ is supported on $\Omega_\nn$ and satisfies \eqref{Neumann_boundary} then it is true (by computation similar to the one in \eqref{discrete_by_parts}) that
$$-\Delta_\Omega^{\Neu}f(x) = -\Delta f(x) \text{ for } x\in \Omega.$$
As the above example shows, using the phrase "Neumann boundary conditions" is misleading, hence we will restrain from using that phrase and use the phrase "Neumann Laplacian" instead.
\end{remark}

Finally we will verify that the operator $-\Delta_\Omega^{\Neu}$ is self-adjoint, meaning that for every $f,g \in L^2(\Omega)$ we have
$$\la f, -\Delta_\Omega^{\Neu} g\ra = \la -\Delta_\Omega^{\Neu} f, g\ra.$$
This follows from the fact that $-\Delta_\Omega^{\Neu}$ is the Laplace operator defined as in \eqref{Laplacian_def} in the previous subsection for the graph $(\Omega, E_\Omega)$, so self-adjointness follows from the general consideration of graph Laplace operators.

\section{The scattering equation on a lattice}
\label{sec:Scattering}

We will start with deriving the formula for the scattering length \eqref{def:scattering_len}. To this end we are interested in a solution to the equation (defined on $\Lambda = A\Z^3$)
\begin{equation}
\label{def:scattering_eq}
-\Delta \varphi(x) + \frac{U}{2}\delta_{x,0} \varphi(x) = 0, 
\end{equation}
with the condition
\begin{equation}
\label{scattering_boundary}
\lim_{|x| \to +\infty} \varphi(x) = 1. 
\end{equation}
This equation is called the (zero-energy) scattering equation. We will see that this equation has a unique solution, hence it is possible to define the scattering length in a following way.
\begin{defi}
\label{def:scattering_len_gen}
The scattering length $\a$ is defined as
\begin{equation*}
4\pi \a = \sum_{x \in \Lambda} \Delta \varphi(x) = \frac{U}{2}\varphi(0),  
\end{equation*}
where $\varphi$ is the solution to the scattering equation \eqref{def:scattering_eq} with condition \eqref{scattering_boundary}.
\end{defi}
In order to solve the scattering equation for the moment we will ignore the condition \eqref{scattering_boundary} and take the (distributional) Fourier transform (see Appendix \ref{sec:Fourier}) of its both sides. A simple computation leads to
\begin{equation}
\label{scattering_transform}
\eps(p) \widehat \varphi + \frac{U}{2}\varphi(0) = 0,
\end{equation}
where $\eps(p)$ is the dispersion relation, defined in \eqref{def:eps}. This equation is satisfied in the sense of distributions, that is after testing against some smooth function on $\Lh$.

For now we will restrict ourselves to the set $\Lh \setminus \{0\}$ and test the above equation with the test function $\phi$ with $\supp \phi$ not including zero. On this set $\left(2\eps(p)\right)^{-1}$ is a well-defined smooth function and therefore we can multiply both sides of the equation \eqref{scattering_transform} by it. It follows that
$$\widehat \varphi = -\frac{U\varphi(0)}{2\eps(p)}.$$
Thus, on this set, we can identify $\widehat \varphi$ as a $L^1(\Lh)$ function (note that this function would not be integrable in the dimensions $d=1$ and $d=2$).

By restricting our considerations to the set not containing zero, we might have neglected distributions whose support is the one-point set $\{0\}$. Since  distributions supported on one point are the sums of Dirac deltas and their derivatives, we conclude that
$$\widehat \varphi = -\frac{U\varphi(0)}{2\eps(p)} + \sum_{\alpha \colon |\alpha| \le M} c_\alpha \partial^\alpha \delta_0,$$
for some $M \ge 0$ and $c_\alpha \in \C$. By equation \eqref{scattering_transform} we need to have
$$\eps(p) \cdot  \left(\sum_{\alpha \colon |\alpha| \le M} c_\alpha \partial^\alpha \delta_0\right) = 0.$$
It follows that $M = 1$ as the value of function $\eps(p)$ and all of its first order derivatives are zero at $p=0$, whereas values of second order derivatives at $p=0$ are non-zero. A consequence of this observation is that
$$\widehat \varphi = -\frac{U\varphi(0)}{2\eps(p)} + C_0 \delta_0 + \sum_{j=1}^3C_j \partial_{p_j}\delta_0.$$
Using the inverse Fourier transform (see equation \eqref{inverse_transform} in the Appendix) we get
$$\varphi(x) = -\frac{U\varphi(0)}{2} |\Lh|^{-1}\int_{\Lh} \frac{e^{ip \cdot x}}{\eps(p)}dp + C_0 + \sum_{j=1}^3 C_j x_j.$$
The value $\varphi(0)$ is not yet specified, we need to make sure that this function is self consistent with its value at $x=0$. Before that we will simplify this expression by using the boundary condition \eqref{scattering_boundary} that so far we have omitted. By the Riemann-Lebesgue lemma we have
$$\lim_{|x| \to \infty}  |\Lh|^{-1} \int_{\Lh} \frac{e^{ip \cdot x}}{\eps(p)}dp = 0,$$
so this part of the scattering equation solution vanishes. An easy observation also leads to conclusion that in order to satisfy \eqref{scattering_boundary} we need to have $C_0 = 1$ and $C_j = 0$ for $j=1,2,3$. We have thus simplified the formula for $\varphi$ to
$$\varphi(x) = 1 - \frac{U\varphi(0)}{2}|\Lh|^{-1}\int_{\Lh} \frac{e^{ip \cdot x}}{\eps(p)}dp.$$
Computing the value at $x=0$ we have
$$\varphi(0) = 1 - \frac{U\varphi(0)}{2}|\Lh|^{-1}\int_{\Lh} \frac{1}{\eps(p)}dp = 1 - U\varphi(0)\gamma,$$
where 
$$\gamma = \frac{1}{2}|\Lh|^{-1}\int_{\Lh} \frac{1}{\eps(p)}dp.$$
This leads to
\begin{equation}
\label{zero_value}
\varphi(0) = \frac{1}{1 + U\gamma} 
\end{equation}
and
\begin{equation}
\label{def:scattering_sol}
\varphi(x) = 1 - \frac{1}{2} \cdot \frac{U}{1 + U\gamma} |\Lh|^{-1}\int_{\Lh} \frac{e^{ip \cdot x}}{\eps(p)}dp.   
\end{equation}
Using \eqref{zero_value} in the Definition \ref{def:scattering_len_gen} we can explicitly write
\begin{equation}
\label{scattering_sol_U}
8\pi \a = \frac{U}{U\gamma + 1},  
\end{equation}
which is the definition used in \eqref{def:scattering_len}


It will also be useful to introduce function $w(x) := 1 -\varphi(x)$ or explicitly
$$w(x) = \frac{1}{2} \cdot \frac{U}{1 + U\gamma} |\Lh|^{-1}\int_{\Lh} \frac{e^{ip \cdot x}}{\eps(p)}dp.$$
This function satisfies the equation
$$\Delta w(x) + \frac12U(1-w(x))\delta_{x,0} = 0$$
with a condition
$$\lim_{|x| \to \infty} w(x) = 0.$$
The main advantage of considering this function instead of $\varphi(x)$ is that its (once again distributional) Fourier transform $\widehat w(p)$ can be treated as a $L^1(\Lh)$ function (and not only as a distribution):
\begin{equation}
\label{w_transform}
\widehat w(p) = \frac{U(1-w(0))}{2\eps(p)} = \frac{U}{1 + U\gamma} \cdot \frac{1}{2\eps(p)}, \quad p \ne 0.  
\end{equation}
We will also note two useful equalities that are frequently used in various parts of the paper:
\begin{equation}
\label{useful}
w(0) = \frac{U\gamma}{1 + U\gamma}, \quad 1 - w(0) = \frac{1}{1 + U\gamma}.
\end{equation}


\medskip
\noindent \textbf{Data availability:} Data sharing is not applicable to this article as no new data were created or analyzed in this study.
\medskip
\\
\noindent \textbf{Funding:} The work of NM and MN was supported by the National Science Centre
(NCN) grant Sonata Bis 13 (project number 2023/50/E/ST1/00439).
\medskip
\\
\noindent \textbf{Declarations}
\\
\noindent
\textbf{Conflicts of interest:} The authors declare no conflict of interest.

\end{document}